\documentclass[english]{smfart}

\usepackage{a4wide}

\usepackage[latin1]{inputenc}
\usepackage[T1]{fontenc}
\usepackage{amssymb}
\usepackage{amsmath}
\usepackage{smfthm}
\usepackage{mathrsfs}
\usepackage{color}
\usepackage{url}
\usepackage{tikz}
\usetikzlibrary{shapes}
\usetikzlibrary{patterns}
\usetikzlibrary{decorations.pathreplacing}
\usepackage[noend]{algpseudocode}
\usepackage{enumerate}
\usepackage{hyperref}
\usepackage{stmaryrd}
\usepackage{booktabs}
\usepackage[capitalise]{cleveref}
\usepackage[ruled,vlined]{algorithm2e}

\usepackage{nicefrac}
\usepackage{thm-restate}

\theoremstyle{plain}
\newtheorem{theorem}{Theorem}
{\bfseries}{\itshape}
{\bfseries}{\itshape}
\newtheorem{proposition}[theorem]{Proposition}
\newtheorem{corollary}[theorem]{Corollary}
\newtheorem{lemma}[theorem]{Lemma}

\newtheorem{pb}{Problem}

\theoremstyle{definition}
\newtheorem{definition}[theorem]{Definition}

\theoremstyle{remark}
\newtheorem{remark}[theorem]{Remark}

\makeatletter
\newcounter{restate}

\makeatother

\makeatletter
\renewcommand*\env@matrix[1][*\c@MaxMatrixCols c]{\hskip -\arraycolsep
  \let\@ifnextchar\new@ifnextchar
  \array{#1}}
\makeatother
 
\newcommand{\cB}{\mathcal{B}}

\newcommand{\cF}{\mathcal{F}}

\newcommand{\cS}{\mathcal{S}}

\newcommand{\Rc}{\mathcal{R}}

\newcommand{\av}{{\mathbf{a}}}
\newcommand{\bv}{{\mathbf{b}}}
\newcommand{\cv}{{\mathbf{c}}}

\newcommand{\vv}{{\mathbf{v}}}
\newcommand{\xv}{{\mathbf{x}}}
\newcommand{\yv}{{\mathbf{y}}}

\newcommand{\Am}{{\mathbf{A}}}
\newcommand{\Bm}{{\mathbf{B}}}
\newcommand{\Cm}{{\mathbf{C}}}
\newcommand{\Dm}{{\mathbf{D}}}
\newcommand{\Em}{{\mathbf{E}}}
\newcommand{\Fm}{{\mathbf{F}}}

\newcommand{\Nm}{{\mathbf{N}}}
\newcommand{\Id}{{\mathbf{I}}}
\newcommand{\Rm}{{\mathbf{R}}}
\newcommand{\Xm}{{\mathbf{X}}}

\newcommand{\Mm}{{\mathbf{M}}}
\renewcommand{\Nm}{{\mathbf{N}}}
\newcommand{\Pm}{{\mathbf{P}}}
\newcommand{\Qm}{{\mathbf{Q}}}
\newcommand{\Sm}{{\mathbf{S}}}
\newcommand{\Tm}{{\mathbf{T}}}
\newcommand{\Um}{{\mathbf{U}}}
\newcommand{\Vm}{{\mathbf{V}}}
\newcommand{\Thetam}{{\mathbf{\Theta}}}
\newcommand{\Gammam}{{\mathbf{\Gamma}}}

\newcommand{\codeVec}[1]{\ensuremath{\mathscr{#1}}_{\textrm{Vec}}}
\newcommand{\Ccv}{{\codeVec C}}
\newcommand{\Dcv}{{\codeVec D}}

\newcommand{\codeMat}[1]{\ensuremath{\mathscr{#1}}_{\textrm{Mat}}}

\newcommand{\codeMatTransp}[1]{\ensuremath{\mathscr{#1}}_{\textrm{Mat}}^{\top}}
\newcommand{\Ccm}{{\codeMat C}}

\newcommand{\Dcm}{{\codeMat D}}

\newcommand{\Dcmt}{{\codeMatTransp D}}

\newcommand{\rk}{\textrm{Rk}}

\newcommand{\Row}{\mathbf{Row}}

\newcommand{\F}{\mathbb{F}}
\newcommand{\Fq}{\F_q}

\newcommand{\Fqm}{\F_{q^m}}

\makeatletter
\newcommand*{\transp}{{\mathpalette\@transpose{}}}
\newcommand*{\@transpose}[2]{\raisebox{\depth}{$\m@th#1\intercal$}}
\makeatother
\newcommand{\transpose}[1]{{#1}^{\top}}

\newcommand{\eqdef}{\mathop{=}\limits^{def}}

\newcommand{\Mspace}[3]{\mathcal{M}_{#1, #2}(#3)}
\newcommand{\Mspdef}{\Mspace{m}{n}{\Fq}}
\newcommand{\sqMspace}[2]{\mathcal{M}_{#1}(#2)}

\newcommand{\GL}[2]{\mathbf{GL}_{#1}(#2)}

\newcommand{\GLk}{\mathbf{GL}_{k}(\Fq)}
\newcommand{\GLn}{\mathbf{GL}_{n}(\Fq)}
\newcommand{\tr}{\textrm{Tr}}

\newcommand{\Mat}[3]{\textup{Mat}_{#1,#2}\left(#3\right)}

\newcommand{\stabl}[1]{\textrm{Stab}_{\textrm{left}}(#1)}

\newcommand{\stabr}[1]{\textrm{Stab}_{\textrm{right}}(#1)}

\newcommand{\cond}[2]{\mathbf{Cond}(#1, #2)}

\def\rad{\textrm{Rad}}

\def\MEP{\textsf{ME}}

\def\MCEP{\textsf{MCE}}
\def\MCREP{\textsf{MCRE}}
\def\VCEP{\textsf{V-MCE}}
\def\hVCEP{\textsf{HV-MCE}}
\def\Falgo{\textsf{F}--algorithm}

\newcommand{\map}[4]{
\left\{
    \begin{array}{ccc}
        #1 & \longrightarrow & #2 \\
        #3 & \longmapsto & #4
    \end{array}
\right.
}
\newcommand{\basis}{\mathcal{B}}
\newcommand{\Mb}{M_{\basis}}
\newcommand{\Mbase}[1]{M_{#1}}
\renewcommand{\leq}{\leqslant}
\renewcommand{\geq}{\geqslant}

\newcommand{\ie}{{\em i.e.}} \title[On the hardness of code equivalence problems in rank metric]{On the Hardness of Code
  Equivalence Problems in Rank Metric}

\author{Alain Couvreur}
\address{Inria}
\address{LIX, CNRS UMR 7161, \'Ecole Polytechnique\\
1 rue Honor\'e d'Estienne d'Orves\\
91120 {\sc Palaiseau Cedex}}
\email{alain.couvreur@inria.fr}

\author{Thomas
  Debris--Alazard}
\address{Inria}
\address{LIX, CNRS UMR 7161, \'Ecole Polytechnique\\
1 rue Honor\'e d'Estienne d'Orves\\
91120 {\sc Palaiseau Cedex}}
\email{thomas.debris@inria.fr}

\author{Philippe Gaborit}
\address{Xlim, CNRS, UMR 7252,
  Universit\'e de Limoges\\
  123, avenue Albert Thomas\\
  87060 Limoges Cedex}
  \email{gaborit@unilim.fr}

\begin{abstract}
  In the recent years, the notion of rank metric in the context of
  coding theory has known many interesting developments in terms of
  applications such as space time coding, network coding or public key
  cryptography.  These applications raised the interest of the
  community for theoretical properties of this type of codes, such as
  the hardness of decoding in rank metric. Among classical problems
  associated to codes for a given metric, the notion of code
  equivalence (to decide if two codes are isometric) has always been
  of the greatest interest, for its cryptographic applications or its
  deep connexions to the graph isomorphism problem.

  In this article, we discuss the hardness of the
  code equivalence problem in rank metric for $\Fqm$--linear and
  general rank metric codes. In the $\Fqm$--linear case, we reduce the
  underlying problem to another one called {\em Matrix Codes Right
    Equivalence Problem}. We prove the latter problem to
  be either in $\mathcal{P}$ or in $\mathcal{ZPP}$ depending of the
  ground field size.  This is obtained by designing an algorithm whose
  principal routines are linear algebra and factoring polynomials over
  finite fields.  It turns out that the most difficult instances
  involve codes with non trivial {\em stabilizer algebras}. The
  resolution of the latter case will involve tools related to finite
  dimensional algebras and Wedderburn--Artin theory. It
  is interesting to note that 30 years ago, an important trend in
  theoretical computer science consisted to design algorithms making
  effective major results of this theory. These algorithmic
  results turn out to be particularly useful in the present article.
  
  Finally, for general matrix codes, we prove that the equivalence
  problem (both left and right) is at least as hard as the
  well--studied {\em Monomial Equivalence Problem} for codes endowed
  with the Hamming metric.
\end{abstract}

\begin{document}
\maketitle

\tableofcontents
\newpage
\section*{Introduction}

\subsection*{The code equivalence problem.} Given two codes over a
finite field, the Code Equivalence problem (CE) asks if they are
isometric when embedded with the Hamming metric, namely if they are
image of each other by a permutation. Code equivalence is a
longstanding problem in computer science with many applications to
cryptography and strong connections with theoretical computer science
problems such as the graph isomorphism problem.  In particular, it has
been proved in \cite{PR97} that the graph isomorphism problem reduces
to the permutation equivalence problem. Therefore any solver of CE in
polynomial time would show that the graph isomorphism problem lies in
$\mathcal{P}$, statement for which we do not know if it is true or not
despite many efforts. This result tends to show that we cannot hope to
get a polynomial time algorithm for solving the code equivalence
problem. On the other hand, \cite{PR97} also proved that this problem
is not $\mathcal{NP}$--complete unless the polynomial-time hierarchy
collapses. Therefore CE lies in an intermediary situation for which
its hardness is not clear.

In another line of works, algorithms have  been considered to
solve the code equivalence problem. The first one, proposed by Leon
\cite{L82}, computes a list of minimum weight codewords
of both codes and try to match them (usually with graph techniques) to
recover the permutation. However such approach leads to an exponential
time algorithm in average over the inputs. Latter Sendrier \cite{S00}
proposed a new approach based on the computation of the hulls (the
intersection of a code and its dual) of both codes.  Interestingly
enough Sendrier gave an algorithm (when codes are over the binary
field) whose ``practical'' complexity is exponential in the dimension
of the hull but since their average dimension for random codes is very
low \cite{S97a} it shows that, in average (and not in the worst case),
that for binary codes (and also over $\F_3$ and $\F_4$), equivalence
is easy to decide. On the other hand, the code equivalence problem
remains difficult for $q \geq 5$ \cite{SS13}.
We emphasize that the efficiency of the aforementioned algorithms
rests on heuristics. They may require an exponential time in
some rare worst cases.

Our contribution in this article is to look at the difficulty of the code equivalence problem but when codes are embedded with the rank metric instead of the Hamming weight.
Note that these kind of equivalence problem exists also with the Euclidean metric where  there is a notion of lattice isomorphism. It has been studied
by Haviv and Regev in \cite{HR14}.

\subsection*{Rank metric and its applications.} Besides the well known notions of Hamming distance for error-correcting codes and Euclidean distance for lattices, there is also the 
concept of rank metric which was introduced in 1951 by Loo-Keng Hu
\cite{H51} as ``arithmetic distance'' for matrices over a field
$\Fq$. Given two $n \times n$ matrices $\Am$ and $\Bm$ over a finite
field $\Fq$, the rank distance between $\Am$ and $\Bm$ is defined as
$|\Am-\Bm| =\textup{Rank}(\Am - \Bm)$. Later, in 1978, Delsarte
defined \cite{D78} the notion of rank distance on the space of
bilinear forms and proposed a
construction of optimal matrix codes in bilinear form
representation. A matrix code over $\Fq$ is defined as an
$\Fq$--linear subspace of the space of $m \times n$ matrices over
$\Fq$ endowed with the rank metric. Later, in 1985, Gabidulin
introduced in \cite{G85} the notion of rank metric codes in vector
representation (as opposed to the matrix representation) over a finite
extension field $\F_{q^{m}}$ of $\Fq$. These codes are known as $\Fqm$-linear codes and they form a particular subclass of matrix codes. In the same paper, Gabidulin
also introduced an optimal class of $\Fqm$-linear codes: the so-called Gabidulin codes, which can be regarded as analogues of Reed-Solomon codes but in a rank metric context, where
polynomials are replaced by so--called {\em linearized} polynomials as
introduced by Ore in 1933 in \cite{O33}. Interestingly, Gabidulin codes are almost the only matrix codes for which we known an efficient decoding algorithm. All matrix codes
with an efficient decoding algorithm are $\Fqm$-linear
\cite{G85,GMRZ13} (there is also simple codes \cite{SKK10} which are not $\Fqm$-linear but they have a trivial structure).

\subsection*{The rank metric and the code equivalence problem.}
The notion of code equivalence with the rank metric has been introduced by Berger in
\cite{B03} (see also \cite{M14}) and invariants with respect to this
code equivalence are considered in \cite{NPH20}.  However, contrary to
the Hamming metric case, to our knowledge, neither the algorithmic
resolution of the code equivalence problem in rank metric nor its
theoretical hardness have ever been discussed in the literature. This is
the purpose of the present paper.

\subsection*{Our contributions}
In this article, we discuss the code equivalence problems in rank metric
from a theoretical and algorithmic perspective. Our contributions in this article are three--fold
\begin{enumerate}[1$^{\circ}$)]
\item We show that the right equivalence problem is
  easy {\em in the worst case}. Interestingly, this proof involves many theoretical and
  algorithmic developments from Wedderburn-Artin theory,
\item The $\Fqm$-linear case is proved to be reduced to the previous
  case and hence is easy in the worst case,
\item Finally, the general case is proved to be harder than the code
  equivalence problem in Hamming metric by providing a polynomial time
  reduction.

\end{enumerate}

It is striking to observe that in rank metric we can get worst case polynomial-time algorithms (possibly Las Vegas) to solve some set of sub-instances of the equivalence problem (points $1$ and $2$) while no such Hamming counterpart seems reachable. In addition these sub-instances are the most considered one in the rank metric based literature.

\subsection*{Outline of the article} The present article is organised
as follows. Section~\ref{sec:tools} recalls main objects and tools
that we consider. Section~\ref{sec:equiv_prob} gives background on the
code equivalence problem in Hamming metric and established various
code equivalence problems in rank metric. Our main results are stated
in Section~\ref{sec:main} and their proofs are sketched; detailed on
proofs are given in the following sections. In
Section~\ref{sec:MCREP}, we show how to solve the {\em right
  equivalence problem}. Section~\ref{sec:equiFqm} considers the code
equivalence problem for $\Fqm$-linear codes and finally, in
Section~\ref{sec:reduction}, we present a polynomial--time reduction
from the code equivalence problem in Hamming metric to the (both left
and right) equivalence problem in rank metric.

\subsection*{Acknowledgements}
The authors express their deep gratitude to Hugues Randriambololona
for his very relevant remarks and to Xavier Caruso for pointing
out references from representation theory.

 \section{Basic objects and tools}\label{sec:tools}

Before explaining the main ideas behind our results, we need to recall
basic objects and tools that we consider.
In all the article, we will consider two different kinds of
objects which we will refer to {\em codes}. Namely:
\begin{itemize}
\item {\em vector codes} are usual codes, {\em i.e.} vector subspaces of
  $\Fq^n$. Such codes are usually endowed with the Hamming metric.
  These codes are denoted with calligraphic letter and with the
  subscript ``Vec'' such as $\Ccv$;
\item {\em matrix codes}, are subspaces of $\Mspdef$ ($m \times n$ matrices whose coefficients belong to $\Fq$) endowed with the
  rank metric. They are denoted as $\Ccm$.
\end{itemize}

Nevertheless, there is class of vector codes that we can equip with
the rank distance. There are subspaces of $\Fqm^{n}$. Such codes are
referred to as {\em rank metric $\Fqm$--linear codes} and there are a
particular class of matrix codes which turn out to be the most
commonly studied in the literature. Let us describe them more
precisely.
Given a vector $\vv \in \Fqm^n$ its {\em rank weight}
is defined as
\begin{equation}
|\vv| \eqdef \textup{Span}_{\F_{q}} \left\{ \lambda_1v_1 + \dots + \lambda_n v_n : \lambda_i \in \Fq \right\}. 
\end{equation}
In other words, $|\vv|$ is the dimension of the $\Fq$--linear subspace of $\Fqm$
spanned by the entries of $\vv$.  
It is well--known that any
$\Fqm$--linear code is isometric to a matrix code as follows. Choose
an $\Fq$--basis $\basis = (b_1, \dots, b_m)$ of $\Fqm$, then consider
the map
\begin{equation}\label{eq:matrix_representation}
	\Mb :  \map{\Fqm^n}{\Mspdef}{(v_1,\dots, v_n)}{
      \begin{pmatrix}
        v_{11} & \cdots & v_{1n}\\
        \vdots &        & \vdots \\
        v_{m1} & \cdots & v_{mn}
		\end{pmatrix}
	}
\end{equation}
where for any $j$, $v_{1j},\dots, v_{mj} \in \Fq$ denote the
coefficients of $v_j$ in the basis $\basis$. That is :
$v_j = v_{1j}b_1 + \cdots + v_{mj}b_m$.  One can easily prove that for
any choice of $\Fq$--basis of $\Fqm$, the above map preserves the
metric.
	
	An {\em $\Fqm$-linear code} of dimension
$k$ is through $\Mb$ a $km$--dimensional matrix code of
$\Mspace{m}{n}{\Fq}$. However, it requires only $k$ vectors in
$\Fqm^n$ or equivalently $k$ matrices in $\Mspace{m}{n}{\Fq}$ to be
represented while a general $km$--dimensional matrix code of
$\Mspace{m}{n}{\Fq}$ is represented by a basis of $km$ matrices of
size $m \times n$. Thus, the representation of an $\Fqm$--linear code
requires $m$ times less memory size compared to that of a general
matrix code of the same $\Fq$--dimension. This gain is particularly
interesting for cryptographic applications. This is basically what
explains why in general McEliece cryptosystems based on rank metric
matrix codes have a smaller key size than McEliece cryptosystems based
on the Hamming metric. All of these proposals (see for instance
\cite{GPT91,GO01,G08,GMRZ13,GRSZ14a,ABDGHRTZABBBO19,AABBBDGZ17,ABGHZ19}) are
actually built from matrix codes over $\Fq$ obtained from
$\Fqm$-linear codes.

On the other hand, vector codes which are $\Fqm$-linear can be viewed
as structured 
matrix codes
with some ``extra'' algebraic structure in the same way as, for instance, cyclic linear codes can be viewed as structured
versions of linear codes. In the latter case, the code is globally
invariant by a linear isometric transform on the codewords
corresponding to shifts of a certain length.

Let
$P = \sum_{i=0}^{m-1}a_i X^i + X^m\in \Fq [X]$ be a monic irreducible
polynomial of degree $m$ and $x \in \Fqm$ be a root of $P$. Then
$\cB \eqdef (1,x,\dots,x^{m-1})$ is an $\Fq$--basis of $\Fqm$.  Let
$\Ccv$ be an $\Fqm$-linear code. The $\Fqm$-linearity with the
definition of $\cB$ means that:
	\[
	\forall \cv\in \Ccv, \forall Q\in \Fq[X], \mbox{ } Q(x) \cdot \cv\in \Ccv.
	\]
	In terms of matrices this stability can be expressed as
    follows. It is readily seen that (see
    \eqref{eq:matrix_representation} for the meaning of $\Mb(\cdot)$)
	\[
	\Mb(P(x)\cv) = P(\Cm_x)\Mb(\cv) \quad \mbox{ where} \quad \Cm_x \eqdef \begin{pmatrix} 
	0 & 0 & \cdots & \cdots &  -a_0 \\
	1 & \ddots & \cdots & \cdots & -a_1 \\
	0 & \ddots & \ddots & & \vdots  \\ 
	\vdots & & \ddots  & \ddots  & \vdots \\
	0 & \cdots & \cdots & 1 & -a_{m-1}
	\end{pmatrix} \in \Fq^{m\times m}
  \]
  {\em i.e.} $\Cm_x$ is the {\em companion matrix} of $P$ which
  represents the multiplication by $x$ in the basis $\mathcal B$.  In
  other words, the matrix code
  $\Mb(\Ccv) \eqdef \left\{ \Mb(\cv) : \cv \in \Ccv \right\}$
  associated to $\Ccv$ is stable by left multiplication by an element
  of the algebra generated by the companion matrix $\Cm_x$ of $x$.

  This property of $\Fqm$-linear codes make them as special a case of
  matrix codes: they have a \textit{large left stabilizer algebra}
  (see \cref{def:stab_alg}). As we will see in \cref{sec:equiFqm},
  the easiness of code equivalence problems for $\Fqm$-linear codes
  comes, roughly speaking, from the fact that they have a particular
  stabilizer algebra.
	
 \section{About the code equivalence problem}\label{sec:equiv_prob}

We arrive now to the core of our paper: the code equivalence
problem. We first recall here the definition of this problem (and some
of its variants) for vector codes endowed with the Hamming metric.

	\subsection{In Hamming Metric}

	 \subsubsection{Statement of the problems}
	
    In Hamming metric, the group of linear isometries of $\Fq^n$ is
    the subgroup of $\GL{n}{\Fq}$ (invertible $n\times n$ matrices over $\Fq$) spanned by permutation matrices and
    nonsingular diagonal matrices. Any element of this group can be
    represented as a product $\Dm \Pm$ where $\Dm$ is diagonal and
    $\Pm$ is a permutation matrix. Such a matrix is usually called a
    {\em monomial matrix}. It is a matrix having exactly one non-zero
    entry per row and per column.

    In Hamming metric, one usually considers two code equivalence problems.
    Given two codes 
endowed with the Hamming
    metric, there are two natural questions:
    \begin{itemize}
    \item are they permutation equivalent (image by each other under a permutation)?
    \item are they monomially equivalent (image by each other under a monomial matrix)?
    \end{itemize}

    \paragraph*{Some comments}
    \begin{enumerate}
      
    \item Note that we only consider linear isometries here.  Some
      references in the literature discuss semi-linear isometries,
      {\em i.e.} the composition of a monomial transform and some iteration
      of a component wise Frobenius map (assuming that the ground
      field $\Fq$ is not prime). This would lead to an alternative problem but that can be solved from any solver of the monomial equivalence problem in essentially the same time. Indeed if one can solve this problem in polynomial time, then it is enough to brute--force any iterate of the
      Frobenius.

  \item Note that MacWilliams equivalence theorem \cite{M62,HP03} asserts that if there
    is a linear map $\Ccv \rightarrow \Dcv$ which is an isometry, then
    it extends to the whole ambient space. Therefore, solving the monomial equivalence problem
    is equivalent to decide whether two codes are image of each other
    by a linear isometry with respect to the Hamming distance.
    \end{enumerate}

    \subsubsection{Algorithms}
    Several algorithms appeared in the literature to solve these problems.
    \begin{itemize}
    \item The first work is probably due to Leon \cite{L82}.
His algorithm computes the full list of minimum weight
      codewords of both codes, which requires an
      exponential running time for almost any code.
    \item Later, Sendrier \cite{S00} introduced the so--called {\em
        support--splitting algorithm} to decide if two codes $\Ccv$ and $\Dcv$ are permutation equivalent. Its approach
      consists in comparing the weight enumerators of the {\em hulls}
      of the codes, {\em i.e.} the codes $\Ccv\cap \Ccv^\perp$ and
      $\Dcv\cap \Dcv^\perp$.  The cost of this algorithm is
      heuristically $O(2^{\dim \Ccv \cap \Ccv^\perp} n^\omega)$
      where $\omega$ denotes the complexity exponent of linear
      algebraic operations. In particular, the algorithm solves easily
      the problem if the hulls of the codes have small dimensions, which
      typically holds \cite{S97a}. However Sendrier's algorithm does not extend to solve the monomial equivalence unless $q =3$ and $q=4$
      if one replaces the hull by its counterpart with respect to the
      Hermitian inner product.
    \item In \cite{F09b} is proposed a notion of normal form of a
      class of codes under the action of the monomial group. No
      mention of complexity appears in this reference but the
      computation of this normal form seems to require an exponential
      running time. In addition, a significant speed-up is possible
      for some input codes
      using tree--cutting dynamical programming methods.

	 \item More recently,  \cite[\S 3.2.4]{S17} proposed to decide if two codes are permutation equivalent to solve a polynomial system with Gr\"obner bases techniques. Actually this
	approach 
can be extended to the
      resolution of the monomial equivalence. However, with this use of
      Gr\"obner bases, the average-time complexity analysis is unknown.
    \end{itemize}

    \subsubsection{Theoretical hardness and known reductions}

    From a more theoretical point of view, what can we say about the
    hardness of these problems? 
Their hardness is not so clear, here is what is known:
    \begin{itemize}
    \item It has been proved in \cite{PR97} that the graph isomorphism
      problem reduces polynomially to the permutation equivalence problem. It is also proved in the
      same reference that this problem is not $\mathcal{NP}$--complete unless the polynomial-time hierarchy
      collapses;
    \item Recently, it was 
      proved \cite{BOS19} that deciding if two codes with zero hull are permutation equivalent reduces 
      to the graph isomorphism problem;
    \item Finally, from \cite{SS13}, the monomial equivalence problem reduces polynomially to
      the permutation equivalence problem when the field cardinality $q$ is polynomial in
      $n$. However, it should be noticed that the reduction sends any
      instance into an instance with a large hull,
      hence in the set of instances which seem to be the hardest ones.
    \end{itemize}

    The reductions are summarized in Figure~\ref{fig:reduction1}.
  
  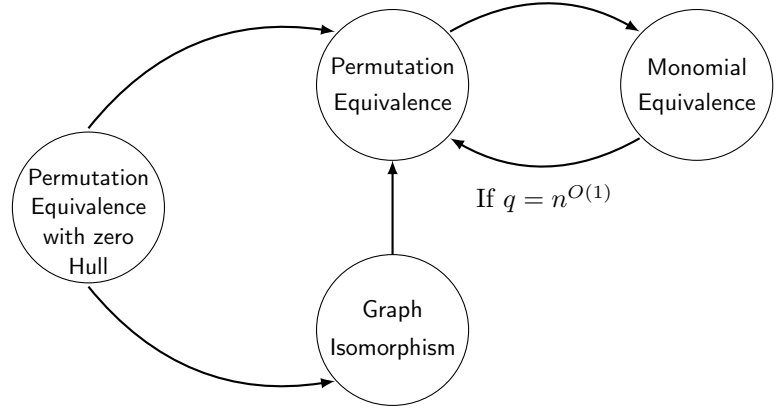
\begin{figure}[h]
    \begin{center}
    	\begin{tikzpicture}
    		\node at (0,0.25) {\scalebox{0.9}{\textsf{Permutation}}};
    		\node at (0,-0.25) {\scalebox{0.9}{\textsf{Equivalence}}};
    		\draw (0,0) circle (1);
    		\node at (0,-3) {\scalebox{0.9}{\textsf{Graph}}};
    		\node at (0,-3.5) {\scalebox{0.9}{\textsf{Isomorphism}}};
    		\draw (0,-3.25) circle (1);
    		\draw[->,thick,>=latex] (0,-2.25) -- (0,-1);
    		\draw (-4,-1.625) circle (1);
    		\node at (-4,-1.225) {\scalebox{0.88}{\textsf{Permutation}}};
    		\node at (-4,-1.625) {\scalebox{0.88}{\textsf{Equivalence}}};
    			\node at (-4,-2) {\scalebox{0.88}{\textsf{with zero}}};
    				\node at (-4,-2.4) {\scalebox{0.88}{\textsf{Hull}}};
    		\draw[->,thick,>=latex] (-4,-0.575)  to[bend left]  (-0.757,0.707);
    			\draw[->,thick,>=latex] (-4,-2.675)  to[bend right]  (-0.757,-3.927);
    		\node at (4,0.25) {\scalebox{0.9}{\textsf{Monomial}}};
    		\node at (4,-0.25) {\scalebox{0.9}{\textsf{Equivalence}}};
    		\draw (4,0) circle (1);
    		\draw[->,thick,>=latex] (0.757,0.707) to[bend left] (3.25,0.707)  ;
    		\draw[<-,thick,>=latex] (0.757,-0.707) to[bend right] (3.25,-0.707);
    		\node at (2,-1.5) {If $q = n^{O(1)}$};
    	\end{tikzpicture}
      \caption{Reductions between various problems. The notation
        ``$\mathsf A \longrightarrow \mathsf B$'' means that ``Problem
        $\mathsf A$ reduces to Problem $\mathsf B$ in polynomial
        time''. It can be translated as ``If one can solve Problem
        $\mathsf B$ in polynomial time, then one can solve Problem
        $\mathsf A$ in polynomial time.}
      \label{fig:reduction1}
	\end{center}
\end{figure}

	\subsection{In Rank Metric}
    In rank metric, the following linear automorphisms of $\Mspace{m}{n}{\Fq}$
    are well--known to preserve the rank:
    \begin{itemize}
    \item left multiplication by an element $\Pm$ of $\GL{m}{\Fq}$ :
      $\Xm \mapsto \Pm \Xm$;
    \item right multiplication by an element $\Qm$ of $\GL{n}{\Fq}$ :
      $\Xm \mapsto \Xm \Qm$;
    \item the transposition map $\Xm \mapsto \Xm^\top$ (only when $m = n$).
    \end{itemize}
    A classical result \cite{H51} of linear algebra
asserts that
    any rank--preserving linear automorphism of $\Mspace{m}{n}{\Fq}$
    is a composition of these three kinds of maps.

    \begin{definition}\label{def:Mat_codes_equiv}
      Two matrix codes $\Ccm,\Dcm\subseteq\Mspace{m}{n}{\Fq}$ are said
      to be {\em equivalent} if there are matrices
      $\Pm \in \GL{m}{\Fq}$ and $\Qm \in \GL{n}{\Fq}$ such that
      $\Ccm = \Pm \Dcm \Qm$.  When $\Pm = \Id_{m}$ (resp.
      $\Qm = \Id_{m}$) codes are said right (resp. left) equivalent.
    \end{definition}

    \begin{remark}
    Note that the previous definition does not involve the possibility
    of a transposition. Hence two square matrix codes
    $\Ccm, \Dcm \subseteq \sqMspace{n}{\Fq}$ are {\em equivalent}
    according to our definition if and only if $\Ccm$ is the image by
    a rank preserving linear map of either $\Dcm$ or $\Dcmt$. In
    practice, if we benefit from an algorithm which decides whether
    two square matrix codes are equivalent (according to our
    definition), then by applying the algorithm successively on the
    pairs $(\Ccm, \Dcm)$ and $(\Ccm, \Dcmt)$, we can decide
    whether the codes are image of each other by a rank preserving
    automorphism. 
    In short, considering a possible transposition will
    only multiply by $2$ the algorithm complexity.
  \end{remark}

  \begin{remark}
    Contrary to the Hamming case there does not seem to exist a rank
    metric McWilliams equivalence theorem \cite[Ex.~2.9]{BG13}. Two
    codes may be image of each other by a linear isometry without
    being equivalent.
  \end{remark}

  In this way, the code equivalence problem for matrix codes endowed
  with the rank metric can be stated as follows.

  \begin{pb}[Matrix Code Equivalence Problem (\MCEP{})]\label{pb:MCEP}
    \mbox{ }
    \begin{itemize}
    \item \textup{Instance:} two matrix codes $\Ccm, \Dcm \subseteq \Mspace{m}{n}{\Fq}$;
    \item \textup{Decision:} there exist matrices $\Pm \in \GL{m}{\Fq}$
      and $\Qm \in \GL{m}{\Fq}$ such that
      \[
        \Ccm = \Pm \Dcm \Qm.
      \]
    \end{itemize}
  \end{pb}

  In addition to this problem, one can be interested in the
  equivalence of the most commonly used rank metric codes, namely
  $\Fqm$--linear codes. If codes are represented with vectors, there
  is no longer equivalence by left multiplying by a non singular
  matrix but right multiplication is still possible. Therefore, 
when regarding vector codes as matrix codes, using the expansion
  operation $\Mb(\cdot)$ defined in (\ref{eq:matrix_representation}),
  equivalence of $\Fqm$-linear codes can be regarded as a restriction of the following problem to
  instances corresponding to matrix representations of $\Fqm$--linear
  codes.

  \begin{pb}[Matrix Codes Right Equivalence Problem (\MCREP{})]
    \begin{itemize}
    \item \textup{Instance:} two matrix codes
      $\Ccm, \Dcm \subseteq \Mspace{m}{n}{\Fq}$;
    \item \textup{Decision:} there exists a matrix $\Qm \in \GL{m}{\Fq}$ such that
      \[
        \Ccm = \Dcm \Qm.
      \]
    \end{itemize}        
  \end{pb}

  Finally, let us introduce a last problem which could be more of
  cryptographic nature. Suppose given two vector codes
  $\Ccv, \Dcv \subseteq \Fqm^n$ and two $\Fq$--bases $\basis, \basis'$
  of $\Fqm$ and consider the matrix codes $\Mbase{\basis}(\Ccv)$ and
  $\Mbase{\basis'}(\Dcv)$. From the data of these matrix codes and
  without knowing the bases $\basis, \basis'$, is deciding equivalence
  easier? This leads to the following problem which is nothing but
  a restriction of \MCEP{} to the subset of instances of matrix codes
  arising from $\Fqm$--linear codes.

  \begin{restatable}[Hidden Vector Matrix Code Equivalence Problem (\hVCEP{})]{pb}{hvCEP}\label{pb:bit_harder}
   	\begin{itemize}
   	\item \textup{Instance:} $\Ccm, \Dcm \subseteq \Mspdef$ be two
   	spaces of matrices representing $\Fqm$--linear codes
   	\item \textup{Decision:} It exists
   	$(\Sm, \Pm) \in \GLk \times \GLn$ such that
   	\[
   	\Ccm = \Sm \Dcm \Pm.
   	\]
   \end{itemize}
  \end{restatable}

 \section{Main results}\label{sec:main}

Our contributions in this article are three--fold
\begin{enumerate}[1$^{\circ}$)]
\item We propose a polynomial time algorithm to solve \MCREP{}. This algorithm
  is deterministic if $q$ is polynomial in $mn$ and Las Vegas for a larger $q$.
\item We prove that \hVCEP{} is easier than \MCEP{} and in particular
  that it naturally reduces to \MCREP{} in polynomial time when $q$ is
  polynomial in $mn$ and via a Las Vegas algorithm for larger $q$.
\item Finally, we prove that the general problem \MCEP{} is at least as hard as
  \MEP{} in Hamming metric by providing a polynomial time reduction.
\end{enumerate}

This leads to the following statements.

\begin{restatable}{theorem}{theoPZPP}\label{theo:PZPP}
  \textup{\MCREP{}} and \textup{\hVCEP{}} are in $\mathcal{P}$ if
  $q = mn^{O(1)}$ and in $\mathcal{ZPP}$ in the general case.
\end{restatable}

\begin{theorem}\label{theo:red}
 The monomial equivalence problem reduces in polynomial time to \textup{\MCEP{}}.
\end{theorem}

\subsection*{The why of $\mathcal{P}$ v.s. $\mathcal{ZPP}$}
In the sequel, we first propose an algorithm to solve \MCREP{}.  Then,
we give a second algorithm which proves that \hVCEP{} reduces to
\MCREP{}.  The major tools of these algorithms are the resolution of
linear systems and factorization of univariate polynomials over a
finite field. Linear systems are solvable in polynomial time while for
factorization of univariate polynomials, Berlekamp algorithm
\cite{B68} is deterministic polynomial only when the ground field
cardinality $q$ is polynomial in the degree. For larger $q$, one
should use a Las Vegas algorithm such as Cantor Zassenhaus algorithm
\cite{B70,CZ81}.  This is the reason why, for a large $q$, we cannot
assert that the problem is in $\mathcal{P}$ but it is in
$\mathcal{ZPP}$.

In order to avoid statements making every time this difference between small and large $q$ we introduce the notion of {\em \Falgo{}}.  
\begin{definition} An {\em \Falgo{}} is an algorithm which uses an oracle (subroutine) to factor polynomials over finite fields. The cost of a call of this oracle is the length of the input of the call.
\end{definition}

In the sequel we sketch the proofs of these theorems. The detailed
proofs appear in the following sections as follows. In Section
\ref{sec:MCREP} we show how to solve \MCREP{} efficiently. Section
\ref{sec:equiFqm} gives an efficient algorithm which reduces \hVCEP{}
to \MCREP{} and therefore proves Theorem \ref{theo:PZPP}. Finally,
\cref{sec:reduction} is devoted to our reduction, namely the proof of
Theorem \ref{theo:red}.

\subsection{A polynomial--time algorithm to solve \MCREP{}}\label{ss:mcrep_sketch}
The resolution of this problem starts with the computation of an
object called the {\em conductor} of $\Ccm$ into $\Dcm$ and
\begin{definition}[Conductor of matrix codes]
  Let $\Ccm, \Dcm \subseteq \Mspace{m}{n}{\Fq}$ and suppose they have
  the same $\Fq$--dimension\footnote{if they do not, then they cannot be
    equivalent.} and define the {\em conductor of $\Ccm$ into $\Dcm$}
  as the $n \times n$ matrix space:
\[
  \cond{\Ccm}{\Dcm} \eqdef \{\Mm \in \sqMspace{n}{\Fq}~|~ \Ccm \Mm
  \subseteq \Dcm\}.
\]
\end{definition}

\begin{proposition}
If $\Ccm \Qm = \Dcm$ for some $\Qm \in \GL{n}{\Fq}$, that is to say,
if $\Ccm$ and $\Dcm$ are right equivalent, then
$\Qm \in \cond{\Ccm}{\Dcm}$.
Conversely, if $\cond{\Ccm}{\Dcm}$ contains
a nonsingular matrix $\Qm$, then $\Ccm$ and $\Dcm$ are right equivalent.  
\end{proposition}

The conductor can be computed by solving a linear system with $n^2$
unknowns and $K(mn - K)$ equations, where
$K = \dim(\Ccm) = \dim( \Dcm)$.  Indeed, our unknowns are matrices
$\Mm \in \cond{\Ccm}{\Dcm}$ whose entries provide $n^2$ unknowns in
$\Fq$.  For the equations, let $\Cm_1, \dots, \Cm_K$ be a basis of
$\Ccm$ and $\Dm_1, \dots, \Dm_{mn-K}$ be a basis of the space
\[
  \{\Mm \in \Mspace{m}{n}{\Fq} ~|~ \forall \Nm \in \Dcm,\
  \tr (\Mm \transpose{\Nm}) = 0\}.
\]
Then, our equations with unknown $\Mm$ are
\[
  \tr (\Cm_i \Mm \Dm_j^\top)= 0,\quad i \in \{1, \dots, K\},\ j \in \{1, \dots,
  N-K\}.
\]
In many classical situations, $m \approx n$, $K$ is linear in $mn$ and
then the number of equations is of $O(n^4)$ while that of variables is
of $O(n^2)$. Thus, the system is over constraint and one can expect
that one of the two following situations appear:
\begin{itemize}
\item either $\cond{\Ccm}{\Dcm} = \{0\}$ and one can
  assert that the codes are not right equivalent;
\item or $\dim(\cond{\Ccm}{\Dcm}) = 1$ and, it
  suffices to pick out one non-zero matrix $\Mm$ in the conductor and
  check if it is nonsingular. If it is non singular, we conclude that
  they are right equivalent, and provide a matrix $\Mm$ which realizes
  the equivalence. If $\Mm$ is singular, since the conductor has
  dimension $1$, any other element of the conductor will be singular
  too and we conclude that the two codes are not equivalent.
\end{itemize}
In both situations, which are the ones that will ``typically'' happen,
\MCREP{} can be solved in polynomial time.

There remains to treat the case where $\cond{\Ccm}{\Dcm}$ has a larger
dimension and may contain both singular and non singular elements. In this
situation, how to find non singular elements? In the other direction,
how to assert that any element in this space is singular?
Answering these questions requires the introduction of a fundamental
notion: the {\em stabilizer algebras} of a matrix code.
\begin{definition}[Stabilizer Algebra]\label{def:stab_alg} Let
  $\Ccm \subseteq \Mspace{m}{n}{\Fq}$ be a matrix code. Its {\em right
    stabilizer} algebra is defined as:
	\[
	\stabr{\Ccm} \eqdef \{\Mm \in \sqMspace{n}{\Fq} ~|~ \Ccm \Mm
	\subseteq \Ccm\}
	\]
	while its {\em left stabilizer} algebra is defined as:
	\[
	\stabl{\Ccm} \eqdef \{\Mm \in \sqMspace{m}{\Fq} ~|~ \Mm \Ccm 
	\subseteq \Ccm\}.
	\]
  \end{definition}

  \begin{remark}
    In terms of terminology, these algebra are also referred to
    as {\em idealisers} or {\em nuclei}. See for instance
    \cite{LN16,LTZ17}.
  \end{remark}

\begin{remark}\label{rem:computing_stabilizers}
  The right stabilizer algebra is nothing but the conductor of $\Ccm$ into
  itself, namely $\cond{\Ccm}{\Ccm}$. In particular, it can be
  computed by solving a system of linear equations.
\end{remark}

It is easily verified that $\stabl{\Ccm}$ and $\stabr{\Ccm}$ are
sub-algebras of $\sqMspace{m}{\Fq}$ and $\sqMspace{n}{\Fq}$
respectively. In particular, for a given code $\Ccm$, its right
stabilizer algebra contains scalar matrices, i.e. matrices of the form
$\lambda \Id_n$. If the stabilizer does not contain other matrices, it
is said to be {\em trivial}. Note that most of the matrix codes have
trivial stabilizer algebras.  The following statement gives a first
motivation for the introduction of these algebras.

\begin{proposition}\label{prop:triv_stab}
  Let $\Ccm, \Dcm \subseteq \Mspace{m}{n}{\Fq}$ be two right
  equivalent codes; i.e. there exists $\Qm \in \GL{n}{\Fq}$ such that
  $\Ccm \Qm = \Dcm$.  Suppose also that
  $\dim(\cond{\Ccm}{\Dcm}) \geq 2$.
  Then, $\Ccm$ and $\Dcm$ have non trivial right stabilizer algebra.
\end{proposition}

\begin{proof}
  By assumption, there exists $\Mm \in \cond{\Ccm}{\Dcm}$ which is not
  a scalar multiple of $\Qm$ and
  $\Mm \Qm^{-1} \in \stabr{\Ccm} \smallsetminus \{\lambda \Id_n ~|~
  \lambda \in \Fq\}$ and
  $\Qm^{-1}\Mm \in \stabr{\Dcm} \smallsetminus \{\lambda \Id_n ~|~
  \lambda \in \Fq\}$.
\end{proof}

While we observed that, when the conductor has dimension $\leq 1$,
then solving \MCREP{} is easy to solve. The previous statement asserts
that, the cases where the conductor has a higher dimension while the
codes are right equivalent are precisely the cases where stabilizer
algebras are non-trivial

In such case, the problem is more technical to prove and requires a
further study of the stabilizer algebra using Wedderburn--Artin
Theory. We refer the reader to \cite{DK94} for the mathematical
aspects of this theory and to the works of Friedl and R\'onyai
\cite{FR85, R90} for the computational aspects.
The results of the previous references assert the existence
of decompositions
\[
  \Ccm = \Ccm \Em_1 \oplus \cdots \oplus \Ccm \Em_s, \qquad
  \Dcm = \Ccm \Fm_1 \oplus \cdots \oplus \Ccm \Fm_t
\]
where the $\Em_i$'s (resp. the $\Fm_i$'s) are projectors matrices onto
subspaces in direct sum. Moreover, these decompositions are unique up
to conjugation and deciding the right equivalence reduces to decide
equivalences of these small pieces as detailed in
\cref{subsubsec:general_mcrep}.

\subsection{Equivalence of $\Fqm$--linear codes}
Clearly, an instance of \VCEP{} is a particular instance of \MCREP{}
and hence, according to the previous discussion, can be solved in
polynomial time. On the other hand, instances of \hVCEP{} are
particular instances of \MCEP{}. Thus, solving \hVCEP{} requires
{\em \`a priori} to look for left and right equivalence at the same time.
However, it turns out, that the $\Fqm$--linear structure permits to treat
these problems separately.

Basically, we proceed as follows. We start from a pair of matrix codes
$\Ccm, \Dcm \subseteq \Mspace{m}{n}{\Fq}$ obtained by expanding vector
codes $\Ccv, \Dcv \subseteq \Fqm^n$ in (possibly distinct)
$\Fq$--bases of $\Fqm$.  The left stabilizer algebras of these codes
can be computed and both are either isomorphic to $\Fqm$ or contain a
subalgebra isomorphic to $\Fqm$.  The latter case is rather more
technical to study (see \cref{ss:cas_penible}), thus let us suppose we
are in the former case. Classical results of linear algebra permit to
prove that two sub-algebras
$\stabl{\Ccm}, \stabl{\Dcm} \subseteq \sqMspace{m}{\Fq}$ which are
both isomorphic to $\Fqm$ are conjugated and that a matrix
$\Pm \in \GL{m}{\Fq}$, that can be efficiently computed, such that
\[
  \stabl{\Dcm} = \Pm^{-1} \stabl{\Ccm} \Pm.
\]
Next, replacing $\Dcm$ by $\Pm \Dcm$, the two codes turn out to have
the same left stabilizer matrix and then, up to some technical details
related to the action of the Frobenius map (see \cref{ss:cas_Fqm_cool}),
we are then reduced to decide right equivalence on the two codes.

\subsection{Reduction from \MCEP{} to \MEP{}}\label{ss:reduc_Hugues} We prove that
the general equivalence problem for matrix codes, \MCEP{}, is at least
as hard as the monomial equivalence problem \MEP{}. We obtain such a
statement by providing a polynomial time reduction from \MEP{} to
\MCEP{}. This reduction was suggested by Hugues Randriambololona.
Another reduction for the search versions of the problems
is given in \cref{sec:reduction}.

Consider the map
\begin{equation}\label{eq:phimap}
  \Phi: \map{\Fq^n}{\sqMspace{n}{\Fq}}{(x_1, \dots, x_n)}{
    \begin{pmatrix}
      x_1 & & \\
      & \ddots & \\
      & & x_n
    \end{pmatrix}.
  }
\end{equation}
An important feature of this map is that it sends vectors with
Hamming weight $t$ on matrices with rank $t$, hence it is an
isometry. The reduction rests on the following statement.

\begin{proposition}
  The map $\Phi$ of \eqref{eq:phimap} sends
  positive (resp. negative) instances of \MEP{} into positive
  (resp. negative) instances of \MCEP{}.
\end{proposition}

\begin{proof}
  Suppose two vector codes $\Ccv, \Dcv \subseteq \Fq^n$ are monomially
  equivalent, \ie{} $\Ccv \Dm \Pm = \Dcv$ where $\Dm$ is a nonsingular
  diagonal matrix and $\Pm$ is a permutation matrix. Then the codes
  $\Ccm \eqdef \Phi(\Ccv)$ and $\Dcm \eqdef \Phi(\Dcv)$ satisfy
  \[
    \Pm^{-1} \Ccm \Dm \Pm = \Dcm.
  \]
Hence, the matrix codes are equivalent.

  Conversely, suppose that $\Ccv, \Dcv$ are not monomially equivalent
  while $\Ccm, \Dcm$ are equivalent.  This entails that $\Ccm, \Dcm$
  are isometric with respect to the rank metric, and, since $\Phi$ is
  an isometry too, the codes $\Ccv, \Dcv$ are isometric with respect
  to the Hamming metric. But from McWilliams equivalence theorem
  \cite{M62,HP03}, this entails the monomial equivalence, a
  contradiction.
\end{proof}

\bigskip

\centerline{***}

 \newpage
\section{Solving the Matrix Codes Right Equivalence Problem (\MCREP{})}\label{sec:MCREP}

This section is devoted to the resolution of \MCREP{}. As we already
noticed in \cref{ss:mcrep_sketch}, most of the time, the resolution of
the problem is simple since the conductor $\cond{\Ccm}{\Dcm}$ has
typically dimension $\leq 1$. If it does not, while the codes are
right equivalent, then, according to \cref{prop:triv_stab}, the codes
$\Ccm$ and $\Dcm$ should have non trivial stabilizer algebras.  The
following statements provide further relations between the structure
of the conductor with that of the stabilizer algebras, which will be
helpful in what follows.

\begin{proposition}\label{prop:conj_stab}
	If two codes $\Ccm$ and $\Dcm$ are right equivalent, i.e.
	$\Ccm \Qm = \Dcm$ for some $\Qm\in \GL{n}{\Fq}$, then
	their right stabilizer algebras are conjugated under $\Qm$:
	\[
	\stabr{\Ccm} = \Qm \cdot \stabr{\Dcm} \cdot \Qm^{-1}.
	\]
\end{proposition}

\begin{proof} Let $\Dm\in\stabr{\Dcm}$. By definition we have the
  following computation,
	\[
		\Ccm\Qm\Dm\Qm^{-1} = \Dcm\Dm\Qm^{-1}
			\subseteq\Dcm\Qm^{-1}
			= \Ccm 
	\]
	which shows that $\Qm\Dm\Qm^{-1}\in\stabr{\Ccm}$ and easily concludes the proof. 
\end{proof} 

This proposition enables to relate the conductor of right equivalent
codes to their stabilizer algebras. It asserts in particular that
$\cond{\Ccm}{\Dcm}$ is a left module on $\stabr{\Ccm}$ and a right
module on $\stabr{\Dcm}$.

\begin{proposition}\label{prop:technical}
  Let $\Ccm, \Dcm \subseteq \Mspace{m}{n}{\Fq}$ be matrix codes and
  $\Qm \in \GL{n}{\Fq}$ such that $\Ccm \Qm = \Dcm$. Then,
  \[
    \cond{\Ccm}{\Dcm} = \stabr{\Ccm} \cdot \Qm = \Qm \cdot \stabr{\Dcm}.
  \]
\end{proposition}

\begin{proof}
  Let $\Sm \in \stabr{\Ccm}$, then
  \[\Ccm \Sm \Qm \subseteq \Ccm \Qm = \Dcm.\]
  Therefore, $\Sm \Qm \in \cond{\Ccm}{\Dcm}$ and hence $\stabr{\Ccm}\Qm
  \subseteq \cond{\Ccm}{\Dcm}$.
  Conversely, let $\Cm \in \cond{\Ccm}{\Dcm}$, then
  \[
    \Ccm \Cm \Qm^{-1}
    \subseteq \Dcm \Qm^{-1} = \Ccm.
  \]
  Therefore, $\Cm \Qm^{-1} \in \stabr{\Ccm}$
  or equivalently $\Cm \in \stabr{\Ccm} \Qm$.
  This yields
  \[
    \cond{\Ccm}{\Dcm} \subseteq \stabr{\Ccm}\Qm.
  \]
  The equality $\cond{\Ccm}{\Dcm} = \Qm \stabr{\Dcm}$ is a consequence of
  the previous equality together with \cref{prop:conj_stab}.
\end{proof}

\subsection{Finite dimensional algebras}\label{subsec:finite_dim_alg}
As said earlier, the general resolution of \MCREP{} involves
results on finite dimensional algebras and the Artin--Wedderburn
theory.
Here we recall some known results in the literature
that we will apply to the right stabilizer algebra of our
codes. We refer the reader to
\cite{DK94} for further details on the theoretical notions and
to \cite{FR85, R90} for the algorithmic aspects.

A {\em finite dimensional algebra} $\mathcal A$ is finite dimensional
vector space over $\Fq$ which is also a ring (possibly non
commutative) with unit. A way to represent such an
algebra is to use a basis $a_1,\dots, a_\ell$ {\em together} with the
collection of so--called {\em structure constants} with respect to
this basis, which are a sequence
$(\lambda_{ijk})_{i,j,k \in \{1,\dots,\ell\}}$ such that:
\[
  \forall i,j \in \{1, \dots, \ell\}, \quad a_i a_j = \sum_{k=1}^\ell
  \lambda_{ijk} a_k.
\]
Any such algebra can be represented as a subring of
$\sqMspace{\ell}{\Fq}$ (\cite[Thm.~1.3.1]{DK94}).

A finite dimensional algebra is {\em simple} if its only two--sided
ideals are $\{0\}$ and the whole algebra itself.  Artin--Wedderburn
theory asserts that any simple algebra over a finite field $\Fq$ is
isomorphic to $\sqMspace{r}{\F_{q^\ell}}$ for some positive $r, \ell$
(see \cite[Cor.~2.4.5 \& Thm.~5.2.4]{DK94}).  A {\em semi--simple}
algebra is a cartesian product of simple algebras.  Another
fundamental object is the {\em Jacobson radical} of an algebra
$\mathcal A$ which is defined as:
\[
  \rad(\mathcal{A}) \eqdef \left\{ x\in\mathcal{A} \mbox{ } : \mbox{ }
    \forall y \in \mathcal{A}, \mbox{ } xy \mbox{ is nilpotent}
  \right\}.
\]
This is a two--sided ideal and, from \cite[Thm.~3.1.1]{DK94},
$\nicefrac{\mathcal A}{\rad{\mathcal A}}$ is semi--simple.  If this
last quotient is a field, $\mathcal A$ is said to be {\em local}.

Finally, another fundamental notion is that of {\em idempotents}. An
element $e \in \mathcal A$ is {\em idempotent} if $e^2 = e$.  In
matrix algebras, idempotent are nothing but projection matrices. Two
idempotents $e_1$ and $e_2$ are said to be {\em orthogonal} if
$e_1e_2 = e_2e_1 = 0$. An idempotent $e$ is said to be {\em minimal}
if $e$ cannot be decomposed as a sum of two non zero orthogonal
idempotents $e = e_1 + e_2$. An algebra is local if and only if its
only idempotent is the unit.

\subsection{General resolution of \MCREP{}}\label{subsec:solveMCREP}

\subsubsection{Context} Let $\Ccm, \Dcm \subseteq \Mspace{m}{n}{\Fq}$
be two matrix codes.
Once $\cond{\Ccm}{\Dcm}$ is computed, our aim is to decide, and
find if exists, a matrix $\Qm \in \GL{n}{\Fq}$ such that
\[
  \cond{\Ccm}{\Dcm} = \stabr{\Ccm} \cdot \Qm = \Qm \cdot \stabr{\Dcm}.
\]
To solve this problem, we first treat the case where $\stabr{\Ccm}$ is
local. Then, we show how to reduce to that case using idempotents of
the stabilizer algebra.

\subsubsection{When the stabilizer algebras are local}\label{ss:local}
It is easy to decide if $\Ccm$ and $\Dcm$ are right equivalent when
$\stabr{\Ccm}$ is a local algebra. Basically it relies on the
following lemma.

\begin{lemma}\label{lemma:loca}
  Let $\Ccm, \Dcm \subseteq \Mspace{m}{n}{\Fq}$ be two matrix codes that are right equivalent and such that
  \(\stabr{\Ccm}\) { is local}. Let $\Rc \eqdef \textup{Rad}(\stabr{\Ccm})$.
Then any element of
  \[\cond{\Ccm}{\Dcm} \smallsetminus \left( \mathcal R \cond{\Ccm}{\Dcm}\right)\] is
  nonsingular.
\end{lemma}

\begin{proof}
  By hypothesis, $\Ccm \Qm = \Dcm$ for some $\Qm \in \GL{n}{\Fq}$.
  From \cref{prop:technical}, $\cond{\Ccm}{\Dcm} = \stabr{\Ccm} \Qm$
  and hence any element of $\cond{\Ccm}{\Dcm}$ can be written
  $\Cm \Qm$ where $\Cm\in\stabr{\Ccm}$. Since $\Qm$ is
  nonsingular, $\Cm \Qm$ is singular if and only if
  $\Cm$ is singular. Finally, since $\stabr{\Ccm}$ is
  local, by \cite[Thm.~3.2.2]{DK94}, $\Cm$ is singular if and only if
  it is in $\mathcal R$.
\end{proof}

This situation yields Algorithm \ref{algo:2} which decides the right-equivalence problem in polynomial time when right stabilizer algebras are local.

	\begin{algorithm}
	\DontPrintSemicolon
	\SetKwInOut{Input}{Input}\SetKwInOut{Output}{Output}
	
	\Input{Two codes $\Ccm, \Dcm \subseteq \Mat{m}{n}{\Fq}$} 
	\Output{Codes are right-equivalent or not}\;
	Compute $\cond{\Ccm}{\Dcm}$, $\mathcal R \eqdef \rad(\stabr{\Ccm})$ and $\mathcal R \cond{\Ccm}{\Dcm}$\;
	Pick $\Am \in \cond{\Ccm}{\Dcm} \smallsetminus \mathcal R \cond{\Ccm}{\Dcm}$\;
	\;
	\If{$\Am$ non singular}{\Return Codes are right-equivalent}
	\Else{\Return Codes are not right-equivalent}
	
	\caption{An algorithm to decide the right equivalence problem when stabilizer algebras are local\label{algo:2}}
\end{algorithm}

\begin{theorem}
	Let $\Ccm, \Dcm$ be matrix codes with local stabilizer algebras. Algorithm \ref{algo:2} succeeds to decide if they are right
	equivalent in polynomial time. 
\end{theorem}

\begin{proof} The correctness of the algorithm is a direct consequence
  of Lemma \ref{lemma:loca}.  For the complexity, the computation of
  the conductor reduces to the resolution of a linear system. The
  computation of the radical can be done in polynomial time thanks to
  \cite{FR85}.
\end{proof}

\subsubsection{The general case}\label{subsubsec:general_mcrep}
In the general case, we aim to reduce the case where stabilizer
algebras of codes are local. For this sake, we are going ``to
decompose'' the unit of the right stabilizer algebras thanks to the
following proposition.

\begin{proposition}\label{propo:decompId}
  Let $\mathcal{A}$ be a $n$-dimensional algebra over $\Fq$. Then, there
  exists a polynomial time \Falgo{}
  computing a
  decomposition of the identity of $\mathcal{A}$ as a sum of minimal
  orthogonal idempotents:
	\[
	1_{\mathcal{A}} = a_1 + \dots + a_{s}.
	\]
	Furthermore this conjugation is unique to conjugation, namely if:
	\[
	1_{\mathcal{A}} = b_1 + \dots + b_t
	\] 
	we have $s = t$ and there exists $g \in \mathcal A^\times$ such that,
	up to re-indexing the $b_i$'s, we have $a_i = gb_ig^{-1}$ for any
	$i \in \{1, \dots, s\}.$
\end{proposition}

\begin{proof} We proceed as follows. 

\smallskip
  
\noindent {\bf Step 1.} Using Frield--R{\'o}nyai algorithm
\cite[Thm~5.7]{FR85}, compute the Jacobson Radical $\rad (\mathcal A)$
of $\mathcal{A}$. This is done in polynomial time;

\smallskip

\noindent {\bf Step 2.}
Compute the semi--simple algebra
$\mathcal S \eqdef\nicefrac{\mathcal{A}}{\rad (\mathcal{A})}$.  From
\cite[Theorem 3.1]{R90} there is a polynomial time \Falgo{} computing
bases of simple algebras $\mathcal{S}_{i}$ such that
	\[
	\mathcal{S} = \mathcal{S}_1 \oplus \dots \oplus \mathcal{S}_{r}.
  \]
  This permits to compute the canonical decomposition of the
  identity into central idempotents:
  \[
    1_{\mathcal S} = 1_{\mathcal S_1} + \cdots + 1_{\mathcal S_r},
  \]
  but these idempotents are not always minimal (unless the $\mathcal S_i$'s
  are all fields) and should be decomposed into minimal ones.  As
  mentioned in Section~\ref{subsec:finite_dim_alg}, since the
  $\mathcal{S}_{i}$'s are simple, they are isomorphic to
  $\sqMspace{u_i}{\F_{q^{v_i}}}$ for some positive
  $u_i,v_i$. Moreover, \cite[Sec.~5.1]{R90} provides a polynomial time
  \Falgo{} to compute these isomorphisms. In the algebra
  $\sqMspace{u_i}{\F_{q^{v_i}}}$ a minimal decomposition of the identity
  is obtained as the sum of the matrices with only one non-zero
  element which is on the diagonal. Using our explicit isomorphism between
  $\mathcal S_i$ and $\sqMspace{u_i}{\F_{q^{v_i}}}$, we get a decomposition of
  of $1_{\mathcal S_i}$ as a sum of minimal orthogonal
  idempotents and deduce such a minimal decomposition for $\mathcal S$:
	\[
	1_{\mathcal S} = e_1 + \cdots + e_s \quad \mbox{where} \quad s \geq r. 
	\]   
\smallskip
    
\noindent {\bf Step 3.}  Since the ground field is finite, from
Wedderburn--Malcev Theorem \cite[Thm.~6.2.1]{DK94}, there exists an
injective morphism of algebras $\mathcal S \hookrightarrow \mathcal A$
and the image of this morphism is unique up to conjugation. Moreover,
such a lift of $\mathcal S$ can be computed by linear algebra as
explained in \cite[\S~1.12]{B11c}.  Using this lift, we deduce a
decomposition of the identity of $\mathcal A$ by minimal orthogonal
idempotents:
	\[
	1_{\mathcal A}= a_1 + \cdots + a_{s}.
	\]
    According to \cite[Thm~3.4.1]{DK94}, this decomposition is
   unique up to conjugation.
\end{proof}

In order to decide the right equivalence of two given matrix codes
$\Ccm, \Dcm$, we start by computing their right stabilizer algebras
 \(\stabr{\Ccm}\) and \(\stabr{\Dcm}\).
Thanks to Proposition~\ref{propo:decompId}, one can compute minimal idempotent
decompositions of the identity of these algebras
\begin{equation}\label{eq:decomp_id_stab}
  \Id_n = \Em_1 + \cdots + \Em_\ell \qquad \text{and} \qquad \Id_n =
  \Fm_1 + \cdots + \Fm_{\ell'}.
\end{equation}
From Proposition~\ref{propo:decompId} the stabilizer algebras are
conjugated if an only if $\ell = \ell'$ and after a suitable
re-indexing, there is a matrix $\Qm$ such that for any $i$,
$\Qm \Em_i \Qm^{-1} = \Fm_i$.
In particular, if $\ell \neq \ell'$ one can stop the process, the codes are
not equivalent. Hence, from now on, we suppose $\ell = \ell'$.
The idea of the general case is to apply ``piecewise'' the algorithm
designed for codes with local algebras to the terms of the following
codes decompositions:
\[
  \Ccm = \Ccm \Em_1 \oplus \cdots \oplus \Ccm \Em_{\ell} \qquad
  \text{and} \qquad \Dcm = \Dcm \Fm_1 \oplus \cdots \oplus \Dcm
  \Fm_{\ell}.
\]
One technical difficulty to tackle is that terms of the decompositions, are matrix codes whose
span of row spaces are not the full space $\Fq^n$. Indeed,
considering for instance the case of $\Ccm$, the matrices $\Em_i$'s
are projectors whose images are in direct sum and any matrix in the space
$\Ccm \Em_i$ has a row space contained in that of $\Em_i$.

To circumvent this issue, we will use the following statement whose proof
is left to the reader.

\begin{lemma}\label{lem:decomp_proj}
  Let $\Em \in \sqMspace{n}{\Fq}$ be a projector of rank $r$
  then there exist two full--rank matrices $\Am, \Bm \in \Mspace{n}{r}{\Fq}$
  such that \[\Em = \Am \transpose{\Bm} \qquad \text{and} \qquad
    \transpose{\Bm}\Am = \Id_r.\]
\end{lemma}

Using this lemma, we decompose the $\Em_i$'s into $\Am_i \transpose{\Bm}_i$
and the $\Fm_i$'s into $\Um_i \transpose{\Vm}_i$ and will focus on the
matrix codes $\Ccm \Am_i$ (resp. $\Dcm \Um_i$) which are contained in $\Mspace{n}{\rk (\Am_i)}{\Fq}$ (resp. $\Mspace{n}{\rk (\Um_i)}{\Fq}$).

\begin{remark}
  Note that replacing $\Ccm$ be a right equivalent code $\Ccm \Qm_0$
  for some $\Qm_0 \in \GL{n}{\Fq}$ and the $\Em_i$'s by their
  conjugates under $\Qm_0$, one could get $\Em_i$'s of the form
  \[
    \Em_i =
    \begin{pmatrix}
    0 &        &   &   &        &   &   &        & \\
      & \ddots &   &   &        &   &   &        & \\
      &        & 0 &   &        &  (0) &   &        & \\
      &        &   & 1 &        &   &   &        & \\
      &        &   &   & \ddots &   &   &        & \\
      &        & (0)  &   &        & 1 &   &        & \\
      &        &   &   &        &   & 0 &        & \\
      &        &   &   &        &   &   & \ddots & \\
      &        &   &   &        &   &   &        & 0
    \end{pmatrix}
  \]
and the code $\Ccm \Am_1$ would be the code $\Ccm$ but keeping only
  the $\rk(\Em_1)$ first columns of each element, the code
  $\Ccm \Am_2$ would be $\Ccm$ but only keeping the $\rk(\Em_2)$
  next columns of each element, and so on.
\end{remark}

Toward the computation of a ``picewise'' right equivalence,
we need the following crucial statements. The first one shows that
right equivalence entails picewise right equivalence.

\begin{proposition}\label{prop:picewise}
  Let $\Ccm, \Dcm$ be two matrix codes in $\Mspace{m}{n}{\Fq}$ such
  that $\Ccm \Qm = \Dcm$ for some $\Qm \in \GL{n}{\Fq}$. Let
  $\Em_1, \dots, \Em_\ell$ and $\Fm_1, \dots, \Fm_\ell$ be minimal
  idempotents of their stabiliser algebras as in (\ref{eq:decomp_id_stab})
  together with their respective decompositions $\Am_i \transpose{\Bm_i}$ and
  $\Um_i \transpose{\Vm_i}$ according to Lemma~\ref{lem:decomp_proj}.
  Then, after a suitable re-indexing, for any $i$ the codes $\Ccm \Am_i$
  and $\Dcm \Am_i$ are right equivalent.
\end{proposition}

\begin{proof}
  Since the $\Fm_i$'s are unique up to conjugation, after possibly replacing
  $\Qm$ by $\Qm \Sm$ for some $\Sm \in \stabr{\Dcm}$ one may assume that
  for any $i$, $\Fm_i = \Qm^{-1}\Em_i\Qm$.
  Therefore,
  \begin{equation}\label{eq:conj_pieces}
    \forall i \in \{1, \dots, \ell\},\quad
    \Ccm \Em_i \Qm = \Ccm \Qm \Qm^{-1}\Em_i \Qm = \Dcm \Fm_i.
  \end{equation}
  Next, recall that $\Em_i = \Am_i \transpose{\Bm_i}$,
  $\Fm_i = \Um_i \transpose{\Vm}_i$ with
  $\transpose{\Vm}_i \Um_i = \transpose{\Bm}_i \Am_i = 0$ and consider
  (\ref{eq:conj_pieces}) again
  \[
    \forall i,\quad \Ccm \Am_i \transpose{\Bm}_i \Qm = \Dcm \Um_i
    \transpose{\Vm}_i \quad
    \Longrightarrow\quad
    \Ccm \Am_i \left(\transpose{\Bm}_i \Qm \Um_i \right) = \Dcm \Um_i
    \underbrace{\transpose{\Vm}_i \Um_i}_{= \Id_{r_i}} = \Dcm \Um_i,
  \]
  where $r_i$ denotes the rank of $\Em_i$.

  There remains to prove
  that $\transpose{\Bm}_i \Qm \Um_i \in \sqMspace{r_i}{\Fq}$ is
  nonsingular. First, note that
  \[
    \rk (\transpose{\Bm}_i \Qm \Um_i) \geq \rk (\Am_i
    \transpose{\Bm}_i \Qm \Um_i \transpose{\Vm}_i) = \rk (\Em_i \Qm \Fm_i).
  \]
  Since $\Fm_i = \Qm^{-1}\Em_i \Qm$, the right hand side becomes
  $\rk (\Em_i^2 \Qm) = \rk(\Em_i \Qm) = r_i$. Therefore
  $\transpose{\Bm}_i \Qm \Um_i$ is a square matrix with full rank, it
  is invertible.
\end{proof}

Conversely, the following statement asserts that picewise right
equivalence entails right equivalence.

\begin{proposition}\label{prop:piecewise_to_global}
  Let $\Ccm, \Dcm$ and the $\Em_i = \Am_i \transpose{\Bm}_i$ and
  $\Fm_i = \Um_i \transpose{\Vm}_i$ as in
  Proposition~\ref{prop:picewise}.  Suppose that for any
  $i \in \{1, \dots, \ell\}$ there exists $\Qm_i \in \GL{r_i}{\Fq}$,
  where $r_i = \textup{Rk}(\Em_i)$, such that $\Ccm \Am_i \Qm_i = \Dcm
  \Um_i$. Then $\Ccm$ and $\Dcm$ are right equivalent. In particular,
  \[
    \Ccm \Qm = \Dcm \qquad \text{where} \qquad \Qm \eqdef \sum_{i=1}^\ell
    \Am_i \Qm_i \transpose{\Vm}_i
  \]
  and the latter matrix $\Qm$ is nonsingular.
\end{proposition}

\begin{proof}
  We have,
  \[
    \Ccm \Qm = \sum_{i=1}^\ell \Ccm \Am_i \Qm_i \transpose{\Vm}_i=
    \sum_{i=1}^\ell \Dcm \Um_i \transpose{\Vm}_i = \Dcm
    \sum_{i=1}^\ell \Fm_i = \Dcm.
  \]
  There remains to prove that $\Qm$ is nonsingular. For this, we will
  prove that
  \[
    \Qm' \eqdef \sum_{i=1}^\ell \Um_i \Qm_i^{-1} \transpose{\Bm}_i
  \]
  is its inverse. Indeed, consider their product:
  \begin{equation}\label{eq:QQ'}
    \Qm \Qm' = \sum_{i=1}^\ell \Am_i \Qm_i \transpose{\Vm}_i \Um_i
    \Qm_i^{-1} \transpose{\Bm}_i + \sum_{i \neq j} \Am_i \Qm_i
    \transpose{\Vm}_i \Um_j \Qm_j^{-1} \transpose{\Bm}_j.
  \end{equation}
  From Lemma~\ref{lem:decomp_proj}, the left--hand sum of
  (\ref{eq:QQ'}) is
  \[
    \sum_{i=1}^\ell \Am_i \transpose{\Bm}_i = \sum_{i=1}^\ell \Em_i = \Id_n.
  \]
  For the right--hand sum of (\ref{eq:QQ'}), note that, again from
  Lemma~\ref{lem:decomp_proj},
  \[
    \forall i \neq j,\quad \transpose{\Vm}_i \Um_j = \transpose{\Vm}_i
    \underbrace{\Um_i \transpose{\Vm}_i}_{=\Em_i} \underbrace{\Um_j
      \transpose{\Vm}_j}_{=\Em_j}\Um_j = 0.
  \]
  Hence, all the terms of the right--hand sum of (\ref{eq:QQ'}) are zero.
  This concludes the proof.
\end{proof}

Finally, we prove that ``pieces'' $\Ccm \Am_i$ have local right
stabiliser algebras. Hence, equivalence can be decided using
Algorithm~\ref{algo:2}.

\begin{proposition}
  Let $\Ccm$, $\Em_i = \Am_i \transpose{\Bm}_i$ as in
  Proposition~\ref{prop:picewise}.  Then, the code $\Ccm \Am_i$
  has a local right stabiliser algebra.
\end{proposition}

\begin{proof}
  Suppose that $\stabr{\Ccm \Am_i}$ is not local, then, there is a
  decomposition of its unit by orthogonal idempotents
  $\Id_{r_i} = \Em_{i_1} + \Em_{i_2}$, where
  $\Em_{i_1}, \Em_{i_2} \in \stabr{\Ccm \Am_i}$ Then,
  \[
    \forall j \in \{1, 2\},\qquad \Ccm \Am_i \Em_{i_j} \transpose{\Bm}_i
    \subseteq \Ccm \Am_i \transpose{\Bm}_i = \Ccm \Em_i \subseteq \Ccm.
  \]
  Therefore, $\Am_i \Em_{i_1} \transpose{\Bm}_i$ and
  $\Am_i \Em_{i_2} \transpose{\Bm}_i$ both lie in
  $\stabr{\Ccm}$. Next, using Lemma~\ref{lem:decomp_proj}, one can
  prove that they are orthogonal idempotents whose sum equals $\Em_i$,
  which contradicts the minimality of $\Em_i$ as an idempotent of
  $\stabr{\Ccm}$.
\end{proof}

	\begin{algorithm}
	\DontPrintSemicolon
	\SetKwInOut{Input}{Input}\SetKwInOut{Output}{Output}
	
	\Input{Two codes $\Ccm, \Dcm \subseteq \Mat{m}{n}{\Fq}$} 
	\Output{Codes are right-equivalent or not}\;
Compute their right stabiliser algebras and minimal decompositions
of their units by orthogonal idempotents:
\[
\Id_n = \Em_1 + \cdots + \Em_\ell \qquad \text{and} \qquad \Id_n =
\Fm_1 + \cdots + \Fm_{\ell'}.
\]\;
\If{$\ell \neq \ell'$}
{\Return Codes are not right equivalent}
\Else{Find a permutation $\sigma\in  \mathfrak S_\ell$ such that  $\forall i \in \{1, \dots, \ell\}$ it exists a non-singular $\Qm_i$ with $\Ccm \Am_i \Qm_i = \Dcm \Um_{\sigma(i)}$ (using
	notation of Proposition~\ref{prop:picewise} ).\;

\If{\textup{Computation fails}}{\Return Codes are not right equivalent}
\Else{\Return Codes are right Equivalent}
}
	
	\caption{An algorithm to solve $\MCREP$ \label{algo:4}}
\end{algorithm}

	The previous material permits to solve \MCREP{} with Algorithm \ref{algo:4}.

\begin{theorem}
	Algorithm \ref{algo:4} is a polynomial time \Falgo{} solving \MCREP{}. 
\end{theorem}

\begin{proof}
	The computation of right stabilizer algebras is done in polynomial time. According to Proposition \ref{propo:decompId} the computation of minimal decompositions of their units by orthogonal idempotents can be done with a polynomial time \Falgo.

	Next step (clearly if $\ell \neq \ell'$ codes are not right equivalent) of the computation can be done in polynomial time using a search version of Algorithm~\ref{algo:2}. Furthermore, if there was a permutation $\sigma \in \mathfrak S_\ell$
	and nonsingular matrices $\Qm_i$ such that
	\[
	\forall i,\quad \Ccm \Am_i \Qm_i = \Dcm \Um_{\sigma(i)},
	\]
	then according to Proposition~\ref{prop:piecewise_to_global}
	the codes are equivalent (and reciprocally) and the equivalence is asserted by the matrix
	$\Qm \eqdef \sum_i \Am_i \Qm_i \transpose{\Vm}_i$ which concludes the proof. 
\end{proof}

 \section{The Code Equivalence Problem for $\Fqm$--linear
  codes : an Easy Problem}
	\label{sec:equiFqm}

    When described as matrix codes, $\Fqm$--linear codes are nothing
    but codes $\Ccm \subseteq \Mspace{m}{n}{\Fqm}$ whose left
    stabilizer algebra $\stabl{\Ccm}\subseteq \sqMspace{m}{\Fq}$
    contains a subalgebra isomorphic to $\Fqm$ as in the following
    definition.

    \begin{definition}[Representation of $\Fqm$]
      A sub-algebra $\cF \subseteq \sqMspace{m}{\Fq}$ is {\em a representation}
      of $\Fqm$ if it satisfies one of the following equivalent conditions:
      \begin{itemize}
      \item It is isomorphic to $\Fqm$ as an $\Fq$--algebra;
      \item It spanned (as an algebra) by a matrix
        $\Am \in \sqMspace{m}{\Fq}$ whose characteristic polynomial is
        $\Fq$--irreducible.
      \end{itemize}  
    \end{definition}

    In particular, taking any monic irreducible polynomial
    $P \in \Fq[X]$ and $\Cm_P$ be its companion matrix, then
    $\Fq[\Cm_P]$ is a representation of $\Fqm$.

    \subsection{Statement of equivalence problems for
    $\Fqm$--linear codes}
  Consider an instance of \cref{pb:bit_harder} (\hVCEP{}) \ie, two
  matrix codes $\Ccm, \Dcm \subseteq \Mspace{m}{n}{\Fq}$ of which are
  known to be matrix representations of $\Fqm$--linear codes but
  possibly represented in two distinct and unknown bases. To reduce to an
  instance of \MCREP{}, our objective is to find a matrix
  $\Pm \in \GL{m}{\Fq}$ such that $\Pm \Ccm$ is right equivalent to
  $\Dcm$.  Such a $\Pm$ will be deduced from the left stabilizer
  algebras of the codes. Recall that these left stabilizer algebras
  $\stabl{\Ccm}, \stabl{\Dcm}$ can be computed by solving a linear
  system in a very similar manner as the computation of right
  stabilizer algebras detailed in \cref{sec:MCREP} (see
  \cref{rem:computing_stabilizers}).  In addition, similarly to
  \cref{prop:conj_stab}, we have that
\[
  \stabl{\Ccm} = \Pm^{-1} \stabl{\Dcm} \Pm.
\]
Thus our objective is, from the knowledge of the stabilizer algebras
$\stabl{\Ccm}$ and $\stabl{\Dcm}$, to solve an algebra conjugacy
problem. It will be possible since these algebras contain a
representation of $\Fqm$, and hence are non-trivial
We treat the problem by considering separately two cases:
\begin{enumerate}[(i)]
\item The left stabilizer algebras of $\Ccm, \Dcm$ are representations
  of $\Fqm$;
\item The left stabilizer algebras of $\Ccm, \Dcm$ contain representations
  of $\Fqm$ as proper sub-algebras
\end{enumerate}

\subsection{When the left stabilizer algebras are representations
of $\Fqm$}\label{ss:cas_Fqm_cool}

In the sequel, we need some very classical results in linear algebra
and representations of $\Fqm$ which are summarized in the following
technical statement.

\begin{lemma}\label{lem:cyclic_morphism}
  Let $\Am \in \sqMspace{m}{\Fq}$ be a matrix whose characteristic
  polynomial $\chi_{\Am}$ is irreducible (for instance the companion
  matrix of any monic irreducible polynomial of degree $m$). Then
  $\Fq[\Am] \simeq \Fqm$ and
  \begin{enumerate}[(i)]
  \item\label{item:commutant} a matrix $\Cm\in \sqMspace{m}{\Fq}$
    commutes with any element of $\Fq[\Am]$ if and only if it is an
    element of $\Fq[\Am]$;
  \item\label{item:transitive} for any
    $\xv, \yv \in \Fq^m \setminus \{0\}$ there exists
    $P(\Am)\in \Fq[\Am]$ such that $P(\Am)\xv^\top = \yv^\top$;
  \item\label{item:Frobenius} There exists
    $\Thetam \in \GL{m}{\Fq}$ such that
    \[
      \Thetam^m = \Id_m \quad {\rm and} \quad \forall P(\Am) \in
      \Fq[\Am],\quad \Thetam P(\Am) \Thetam^{-1} = P(\Am)^q.
    \]
  \item\label{item:2nd_Frobenius}
    Any matrix $\Gammam \in \GL{m}{\Fq}$
    satisfying $\Gammam \Fq[\Am] \Gammam^{-1} = \Fq[\Am]$
    is of the form $\Gammam = \Thetam^j P(\Am)$
    for some $P(\Am) \in \Fq[\Am]^\times$ and some
    $j \in \{0, \dots, m-1\}$.
  \item\label{item:conjugate} any matrix $\Bm \in \sqMspace{m}{\Fq}$
    with the same characteristic polynomial as $\Am$, then $\Am$ and
    $\Bm$ are conjugated, {\em i.e.}  there exists
    $\Pm\in \GL{m}{\Fq}$ such that $\Am = \Pm^{-1}\Bm\Pm$ and $\Pm$
    can be computed in polynomial time.
  \end{enumerate}
\end{lemma}

\begin{proof}
  Clearly $\Fq[\Am]$ is commutative. Conversely, if $\Cm$
    commutes with any element of $\Fq[\Am]$ then it commutes with
    $\Am$.  Since the characteristic polynomial $\chi_{\Am}$ of $\Am$
    is irreducible and $\Fq$ is perfect, then $\Am$ has distinct
    eigenvalues and it is a well--known fact in linear algebra that
    this latter features entails that $\Cm$ should be a polynomial in $\Am$.
    This proves (\ref{item:commutant}).

    \medskip

    Consider the space $F \eqdef \{P(\Am) \xv^\top ~|~ P \in
    \Fq[X]\} \subseteq \Fq^m$. This space is isomorphic to
    $\nicefrac{\Fq[X]}{(\pi_{\Am, \xv})}$ where $\pi_{\Am, \xv}$ is
    the monic polynomial $P$ of lowest degree satisfying
    $P(\Am)\xv^\top = 0$. This polynomial divides the minimal polynomial of
    $\Am$ which itself is known to divide $\chi_{\Am}$. Since, by hypothesis
    $\chi_{\Am}$ is irreducible, then $\pi_{\Am, \xv} = \chi_{\Am}$ and
    $\dim_{\Fq} F = \dim_{\Fq} \nicefrac{\Fq[X]}{(\chi_{\Am})} = m$.
    Therefore, $F = \Fq^m$ and hence contains $\yv$, which
    proves (\ref{item:transitive}).

    \medskip
    
    Using the fact that $\Fq^m$ is isomorphic as a vector space
    with $\Fqm$, the matrices in $\Fq[A]$ represent the multiplications
    by elements of $\Fqm$ in some given basis. The matrix $\Thetam$
    is nothing but a matrix representation of the
    Frobenius map $x \mapsto x^q$
    in this basis. This proves (\ref{item:Frobenius}).

    \medskip
    
    Let $\Gammam$ be a matrix such that $\Fq[\Am]$ is globally
    invariant by conjugation under $\Gammam$.
    This conjugation map induces a non-trivial automorphism of the field
    $\Fq[\Am] = \Fqm$, which is nothing but an iterate of the Frobenius.
    Therefore,
    \[
      \forall P(\Am) \in \Fq[\Am], \quad
      \Gammam P(\Am) \Gammam^{-1} = P(\Am)^{q^j}
    \]
    for some $j \in \{0, \dots, m-1\}$.
    Now, ${\Gammam}^{-1} \Thetam^j$ commutes
    with any element of $\Fq[\Am]$ and then (\ref{item:2nd_Frobenius})
    can be deduced from (\ref{item:commutant}).
    
    \medskip
    
    Let $\xv \in \Fq^m \setminus \{0\}$, reasoning
    as in the proof of (\ref{item:transitive}) we see that
    $(\xv^\top, \Am \xv^\top, \dots, \Am^{m-1}\xv^\top)$ is a basis of
    $\Fq^m$. Let $\Pm_0$ be the transition matrix from the canonical
    basis to this basis. This matrix can be computed in polynomial time
    and Then, $\Pm_0^{-1} \Am \Pm_0$ is nothing but
    the companion matrix of $\chi_\Am$. Similarly, one can compute
    $\Pm_1$ such that $\Pm_1^{-1} \Bm \Pm_1$ equals this companion
    matrix of $\chi_{\Am}$. The matrix $\Pm \eqdef \Pm_1\Pm_0^{-1}$
    yields ($\ref{item:conjugate}$).
\end{proof}

This technical lemma has the following consequence on
the matrix representations of $\Fqm$.

\begin{corollary}\label{cor:conjug_algebras}
  Let $\cS_1, \cS_2 \subseteq \sqMspace{m}{\Fq}$ be two matrix
  representations of $\Fqm$. Then there exists a matrix
  $\Pm \in \GL{m}{\Fq}$ such that
  $\mathcal{S}_1 = \Pm^{-1} \mathcal{S}_2 \Pm$.
  In addition, $\Pm$ can be computed using a polynomial time
  \Falgo{}.
\end{corollary}

\begin{proof}
  First, compute matrices $\Am_1, \Am_2 \in \sqMspace{m}{\Fq}$ such
  that $\cS_1 = \Fq[\Am_1]$ and $\cS_2 = \Fq[\Am_2]$.  Such matrices
  can be computed in polynomial time.
Denoting by $\chi_{\Am_1}, \chi_{\Am_2}$ their respective
  characteristic polynomials we have
  \[
    \cS_1= \Fq[\Am_1] \simeq \nicefrac{\Fq[X]}{(\chi_{\Am_1})}
    \quad {\rm and} \quad
    \cS_2 = \Fq[\Am_2] \simeq \nicefrac{\Fq[X]}{(\chi_{\Am_2})}.
  \]
  Using an \Falgo{}, one can compute a root of $\chi_{\Am_2}$ in
  $\Fq[\Am_1] \simeq \nicefrac{\Fq[X]}{\chi_{\Am_1}}$, and denote it by
  $P(\Am_1)$. Then $P(\Am_1)$ has the same characteristic polynomial
  as $\Am_2$ and, from
  \cref{lem:cyclic_morphism}($\ref{item:conjugate}$), there is a
  matrix $\Pm$ which can be computed in polynomial time such that,
  $\Am_2 = \Pm^{-1} P(\Am_1)\Pm$. Moreover, for dimensional reasons,
  $\Fq[\Am_1] = \Fq[P(\Am_1)]$ and hence
  \[
    \cS_2 = \Fq[\Am_2] = \Pm^{-1} \Fq[\Am_1] \Pm = \Pm^{-1} \cS_1 \Pm.
  \]
\end{proof}

So, the first part of our algorithm will consist in computing
the left stabilizer algebras of ${\Ccm}$ and ${\Dcm}$ (which are
assumed to be representations of $\Fqm$). Next, according
to Corollary~\ref{cor:conjug_algebras}, the stabilizer
algebras are conjugated and one can compute (in polynomial time if $q = m^{O(1)}$ otherwise with a Las Vegas method) $\Pm \in \GL{m}{\Fq}$ such that:
\[
  \stabl{\Dcm} = \Pm^{-1} \stabl{\Ccm} \Pm.
\]
Replacing $\Dcm$ by the left equivalent code $\Pm \Dcm$, then the
codes turn out to have the same left stabilizer algebra.
Once we are reduced to the case where the two codes have the same
left stabilizer algebra, the following statement asserts
that we are almost reduced to solving and instance of \MCREP{}.

\begin{proposition}\label{prop:Frob_left_equivalence}
  Let $\Ccm, \Dcm \subseteq \Mspace{m}{n}{\Fq}$ be two left equivalent
  matrix codes codes whose left stabilizer algebras are equal and are
  representations of $\Fqm$, then there exists
  $j \in \{0, \dots, m-1\}$ such that
  \[
    \Ccm = \Thetam^j \Dcm
  \]
  where $\Thetam$ is a matrix defined in
  \ref{lem:cyclic_morphism}~(\ref{item:Frobenius}).
\end{proposition}

\begin{proof}
  The codes are suppose to be left equivalent, then
  \( \Ccm = \Sm \Dcm \) for some $\Sm \in \GL{n}{\Fq}$. In addition,
  if they have the same left stabilizer algebra $\cS$,
  which is supposed to be of the form $\Fq[\Am] \simeq \Fqm$
  for some $\Am \in \sqMspace{m}{\Fq}$ with an irreducible characteristic
  polynomial. Then, from a ``left version'' of
  Proposition \ref{prop:conj_stab}, $\cS$ is stable by conjugation by $\Sm$. Thus,
  from \cref{lem:cyclic_morphism}($\ref{item:2nd_Frobenius}$),
  we have
  \[
    \Sm = \Thetam^j P(\Am)
  \]
  for some $j \in \{0, \dots, m-1\}$ and $P(\Am) \in \Fq[\Am]^\times$.
  Therefore, since $P(\Am) \in \cS = \stabl{\Dcm}$, then $\Sm \Dcm=
  \Thetam^j \Dcm$, which concludes the proof.
\end{proof}

\begin{corollary}\label{cor:frob_left_equivalence}
  Let $\Ccm, \Dcm \subseteq \Mspace{m}{n}{\Fq}$ which are equivalent,
  {\em i.e.} there exists $\Sm \in \GL{m}{\Fq}$ and
  $\Tm \in \GL{n}{\Fq}$ such that $\Ccm = \Sm \Dcm \Tm$. Suppose that
  the two codes have the same left stabilizer algebra which is a
  representation of $\Fqm$. Then, denoting by $\Thetam$ the matrix of
  \cref{lem:cyclic_morphism}~($\ref{item:Frobenius}$), $\Ccm$ is right
  equivalent to $\Thetam^j \Dcm$ for some $j \in \{0, \dots, m-1\}$.
\end{corollary}

\begin{proof}
  Apply \cref{prop:Frob_left_equivalence} to the pair $(\Ccm, \Dcm \Tm)$.
\end{proof}

These results yield to Algorithm \ref{algo:3} which decides \hVCEP{} in polynomial efficiently when left stabilizer algebras are equal to $\Fqm$.

\begin{algorithm}
	\DontPrintSemicolon
	\SetKwInOut{Input}{Input}\SetKwInOut{Output}{Output}
	
	\Input{Two codes $\Ccm, \Dcm \subseteq \Mat{m}{n}{\Fq}$} 
	\Output{Codes are equivalent or not}\;
	Compute  the left stabilizer algebras of $\Ccm, \Dcm$ and $\Pm$
	such that they are conjugated under $\Pm$;\;
	$\Dcm \leftarrow \Pm \Dcm$ \;

Compute a matrix $\Thetam$ as in \cref{lem:cyclic_morphism}(\ref{item:Frobenius}) \;
	\For{$j \in \{1,\cdots,m\}$}{
	{Compute $\Thetam^j \Dcm$\;}
	\If{$\Ccm$ and $\Thetam^t \Dcm$ are right equivalent}{\Return Codes are equivalent}}
\Return Codes are not equivalent
	\caption{An algorithm to decide the equivalence problem for $\Fqm$-linear code with left stabilizer algebras equal to $\Fqm$.\label{algo:3}}
\end{algorithm}

\begin{theorem}
  Let $\Ccm, \Dcm$ be two spaces of matrices representing
  $\Fqm$-linear codes and whose left stabilizer algebra is a
  representation of $\Fqm$. Algorithm \ref{algo:3} is a polynomial time
  \Falgo{} which succeeds to decide
  if they are equivalent.
\end{theorem}

\begin{proof} The {\bf for} loop is made at most $m$ times, which is
  polynomial in the size of the input.  According to
  Proposition~\ref{cor:conjug_algebras}, there is a polynomial time
  \Falgo{} computing $\Pm$. Next, using the polynomial time \Falgo{}
  described in \cref{subsec:solveMCREP} one can decide
  the right equivalence of $\Ccm$ and $\Thetam^t \Dcm$.
\end{proof}

\subsection{When the left stabilizer algebras strictly contain
representations of $\Fqm$}\label{ss:cas_penible}
In the general case, it is possible that the left stabilizer algebras
contain representations of $\Fqm$ as proper subspaces.
To understand the way to proceed, we first have to classify
sub-algebras of $\sqMspace{m}{\Fq}$ containing a representation
of $\Fqm$,
which is the point of the next statement.

\begin{proposition}\label{prop:key_prop_for_general_fqm-lin}
  Let $\cS \subseteq \sqMspace{m}{\Fq}$ be an algebra such that
  \[
    \Fq[\Am] \simeq  \Fqm \varsubsetneq \cS \subseteq \sqMspace{m}{\Fq},
  \]
  then, $\cS$ is isomorphic to a matrix algebra of the form
  $\sqMspace{\frac m \ell}{\F_{q^{\ell}}}$ for some $\ell$
  dividing $m$. In addition, any $\Cm \in \sqMspace{m}{\Fq}$
  which commutes with any element of the centre of $\cS$, then $\Cm
  \in \cS$.
\end{proposition}

\begin{proof}
  First, let us prove that $\cS$ is semi-simple, namely that its
  Jacobson radical is trivial. Indeed, let
  $\Nm \in \cS \setminus \{0\}$ be such that for any $\Mm \in \cS$,
  $\Nm \Mm$ is nilpotent. Then, there exists $\xv \in \Fqm$ such that
  $\Nm \xv^\top \neq 0$ and, from
  Lemma~\ref{lem:cyclic_morphism}~(\ref{item:transitive}), there
  exists $P(\Am) \in \Fq[\Am]$ such that
  $P(\Am)\Nm \xv^\top = \xv^\top$.  Thus, $P(\Am)\Nm \in \cS$ and has
  a non-zero fixed point which contradicts its nilpotence. Therefore,
  $\rad (\cS) = \{0\}$ and hence $\cS$ is semi-simple.

  Now let us prove that $\cS$ is simple. For that,
  using \cite[Cor.~1.7.8]{DK94}, we only need to prove that its centre
  is simple. But, since $\cS$ contains $\Fq[\Am]$, from
  \cref{lem:cyclic_morphism}~(\ref{item:commutant}), its centre
  is contained in $\Fq[\Am]\simeq \Fqm$ and hence is isomorphic
  to $\F_{q^\ell}$ for some $\ell$ dividing $m$. In particular, it
  is simple and so is $\cS$. Therefore, $\cS$ is isomorphic
  to $\sqMspace{r}{\F_{q^\ell}}$ for some positive $r$ and there
  remains to prove that $r = \frac m \ell$.

  From \cref{lem:cyclic_morphism}($\ref{item:commutant}$) any
  element $\Fq[\Am] \simeq \Fqm$ in $\cS$ commuting with any element
  of $\Fq[\Am]$ is in $\Fq[\Am]$. Using \cite[Thm.~4.4.6(3)]{DK94} we
  deduce that the $\F_{q^\ell}$--dimension of $\cS$ is
  $(\frac{m}{\ell})^2$.  This proves that $r = \frac m \ell$.
  The former reference permits also to prove that a
  matrix in $\sqMspace{m}{\Fq}$ commuting with any element of $\cS$
  is in $\cS$.
\end{proof}

Before describing our algorithm to decide equivalence when the left
stabilizer algebras strictly contain a representation of $\Fqm$, we
need the following statement, which is an analogue of
\ref{cor:frob_left_equivalence} in the current situation. Note that
the current situation is in some sense more favorable since we no
longer need to take care of a possible action of the Frobenius
morphism.

\begin{proposition}\label{prop:same_big_stab}
  Let $\Ccm, \Dcm \subseteq \Mspace{m}{n}{\Fq}$ be two left
  equalvalent matrix codes, \ie{} such that
  \[
    \exists \Pm \in \GL{m}{\Fq},\ \Pm \Ccm = \Dcm.
  \]
  Suppose also that
  \[
    \stabl{\Ccm} = \stabl{\Dcm} \simeq \sqMspace{\frac m \ell}{\F_{q^\ell}},
  \]
  for some positive $\ell$ dividing $m$.
  Then, $\Ccm = \Dcm$.
\end{proposition}

\begin{proof}
  Let $\cS \eqdef \stabl{\Ccm} = \stabl{\Dcm}$
  Using a left version of \ref{prop:conj_stab}, we deduce that
  $\cS$ is stable by conjugation by $\Pm$. Since $\cS$ is a central simple
  algebra over $\F_{q^\ell}$, from
  Skolem--Noether Theorem \cite[Cor.~4.4.3]{DK94}, its automorphisms are inner
  and hence there exists $\Rm \in \cS$ such that
  \[
    \forall \Sm \in \cS,\ \Pm^{-1}\Sm \Pm = \Rm^{-1}\Sm \Rm.
  \]
  Therefore, the matrix $\Rm\Pm^{-1}$ is in the {\em centraliser} of
  $\cS$ in $\sqMspace{m}{\Fq}$, (\ie{} the subalgebra of
  $\sqMspace{m}{\Fq}$ of matrices commuting with any element of
  $\cS$).  The centraliser contains the centre of $\cS$, which is
  isomorphic to $\F_{q^\ell}$.  Using \cite[Thm.~4.4.6(3)]{DK94}, we
  deduce that the centraliser and the centre have the same dimension
  and hence are equal. Consequently $\Rm\Pm^{-1}$ is in $\cS = \stabl{\Ccm}$ and
  hence so is $\Pm$. Thus, $\Ccm = \Pm \Ccm = \Dcm$.
\end{proof}

With \cref{prop:key_prop_for_general_fqm-lin} in hand,
we can deduce an algorithm for solving this general case. We describe it as follows.

\begin{enumerate}
\item Compute the left stabilizer algebras $\stabl{\Ccm}, \stabl{\Dcm}$;
\item Compute their centres
  $Z(\stabl{\Ccm}), Z(\stabl{\Dcm})$ which are isometric to a same $\F_{q^\ell}$ for some $\ell > 0$, otherwise codes are not equivalent;
\item These algebras are respectively isomorphic to
  $\Fq[\Am_1], \Fq[\Am_2]$ for some
  $\Am_1, \Am_2 \in \sqMspace{m}{\Fq}$ whose minimal polynomial is
  irreducible of degree $\ell$. A very similar argument as
  \cref{cor:frob_left_equivalence} permits to assert that the algebras
  are conjugated and one can compute a matrix $\Pm \in \GL{m}{\Fq}$
  such that $\Fq[\Am_1] = \Pm \Fq[\Am_2] \Pm^{-1}$.
\item Replacing $\Dcm$ by $\Pm \Dcm$ so that the stabilizer algebras
  of the codes have the same centre $Z$. But, from
  \cref{prop:key_prop_for_general_fqm-lin}, the stabilizer algebras are
  then the same since they both equal to the commutator of $Z$ in
  $\sqMspace{m}{\Fq}$.
\item Now, since the two codes have the same left stabiliser algebra,
  from \cref{prop:same_big_stab}, they are equivalent if and only if
  they are right equivalent and we can decide this right equivalence
  using the algorithm of Section~\ref{sec:MCREP}.
\end{enumerate}

This enables to prove the following theorem.

\begin{theorem}
	\hVCEP{} is in $\mathcal{P}$ if $q = (mn)^{O(1)}$ and in $\mathcal{ZPP}$ in the general case. 
\end{theorem}

 \section{Reduction from the Code Equivalence Problem in
Hamming Metric}\label{sec:reduction}
The objective of the present section is to prove that the Matrix Code
Equivalence Problem \MCEP{} (\cref{pb:MCEP}) for rank metric codes is
at least as hard as the Monomial Equivalence Problem. Compared to the
reduction presented in \cref{ss:reduc_Hugues}, the following one
treats search versions of the problem and is somehow more explicit.

From
the generator matrices point of view, to decide if two codes are monomially equivalent may be
formulated as:

\begin{itemize}
\item \textup{Instance :} Two matrices $\Am,\Bm\in\Mspace{k}{n}{\Fq}$.
\item \textup{Decision :} it exists $\Sm\in\GL{k}{\Fq}$,
  $\Pm\in\sqMspace{n}{\Fq}$ be a permutation matrix and
  a diagonal matrix $\Dm \in \GL{n}{\Fq}$ such that:
  \[
  \Am = \Sm\Bm\Dm\Pm. 
  \]
\end{itemize}

Without loss of generality we can assume that the columns of $\Am$ (resp.
$\Bm$) are pairwise linearly independent.
If not, we can remove the redundant columns
without changing the answer of the problem.

Let
$\av_1^\top,\cdots,\av_n^\top, \bv_1^\top,\cdots,\bv_n^\top \in \Fq^k$
denote the columns of $\Am$ and $\Bm$ respectively. Similarly to the
previous section, we keep on considering that vectors are row
matrices, which explains why we apply the transposition operator to
these column vectors.
For any vector $\xv\in\Fq^{k}$ and any
$1\leq i \leq n$, we define the $n\times k$ matrix $\Row_{i}(\xv)$ over $\Fq$
whose only non-zero row is the $i$--th one which equals $\xv$:

\begin{center}
	\begin{tikzpicture}
	\node at (0,0) {$\Row_{i}(\xv) \eqdef \begin{pmatrix}
		0 & \cdots & 0 \\ 
		& \xv & \\
		0 & \cdots & 0 \\ 
		\end{pmatrix}$};
	\draw[<-] (1.7,0) -- (2.25,0);
	\node at (2.5,0) {$i$.};
\end{tikzpicture} 
\end{center}
We now build in polynomial time in the size of $\Am$ and $\Bm$ the
following matrix codes of $\Mspace{k+n}{k}{\Fq}$:
\begin{align*}
\Ccm &\eqdef \left\{ \sum_{i=1}^{n}\lambda_i\begin{pmatrix} \transpose{\av}_i \av_{i} \\
\Row_{i}({\av}_i)\\
\end{pmatrix}  : \lambda_i \in \Fq \right\}, \\
\Dcm &\eqdef \left\{\sum_{i=1}^{n}\lambda_i\begin{pmatrix} \transpose{\bv}_i\bv_{i} \\
  \Row_{i}({\bv}_i) \\
\end{pmatrix} : \lambda_i \in \Fq \right\}.
\end{align*} 
The following lemma is crucial for our reduction, it justifies
our construction of $\Ccm$ and $\Dcm$.

\begin{lemma}\label{lemma:int} Let $\Um\in\GL{k+n}{\Fq}$ and
  $\Vm\in \GL{k}{\Fq}$ which verify:
	\[
	\Ccm = \Um\Dcm\Vm. 
	\]
	Then, there exists a permutation $\sigma \in \mathfrak S_n$
    and $\alpha_1, \dots, \alpha_n \in \Fq^\times$ such that:
	\[
	\begin{pmatrix} \transpose{\av}_i\av_{i} \\
      \Row_{i}({\av}_i)\\
	\end{pmatrix} =  \alpha_{\sigma (i)} \Um \begin{pmatrix} \transpose{\bv}_{\sigma (i)} \bv_{\sigma (i)} \\
      \Row_{\sigma (i)}({\bv}_{\sigma (i)}) \\
	\end{pmatrix} \Vm.
	\]
\end{lemma}

\begin{proof}
  By definition, for any $1 \leq i \leq n$ there exist
  $\alpha_1,\dots,\alpha_{n} \in \Fq$ such that:
	\[
	\begin{pmatrix} \transpose{\av}_i\av_{i} \\
	\Row_{i}({\av}_i)\\
  \end{pmatrix} = \sum_{j=1}^{n} \alpha_{j} \Um \begin{pmatrix} \transpose{\bv}_j\bv_{j} \\
	\Row_{j}({\bv}_j) \\
	\end{pmatrix} \Vm = \Um\left( \sum_{j=1}^{n} \alpha_{j} \begin{pmatrix} \transpose{\bv}_j \bv_{j} \\
	\Row_{j}({\bv}_j) \\
	\end{pmatrix}\right) \Vm.
    \]
	But now, since $\Um$ and $\Vm$ are non-singular,
	\[
	\rk\left( \Um\left( \sum_{j=1}^{n} \alpha_{j} \begin{pmatrix} \transpose{\bv}_i\bv_{i} \\
	\Row_{i}({\bv}_i) \\
	\end{pmatrix}\right) \Vm\right) =  \rk\left( \sum_{j=1}^{n} \alpha_{j} \begin{pmatrix} \transpose{\bv}_j\bv_{j} \\
	\Row_{j}({\bv}_j) \\
	\end{pmatrix}\right).
	\]	
	Recall that the rows of $\Row_{j}(\bv_j)$ are all zero but the
    $j$--th one which equals $\bv_{j}$. Here, by hypothesis, the
    $\bv_{j}$'s are pairwise linearly independent.  Therefore, if
    there were at least two non zero $\alpha_j$'s, then we would have:
	\[
	\rk\left( \sum_{j=1}^{n} \alpha_{j} \begin{pmatrix} \bv_{j}\transpose{\bv}_j \\
	\Row_{j}(\transpose{\bv}_j) \\
	\end{pmatrix}\right) \geq 2. 
\]
Indeed, the considered matrix would have two independent rows in its
lower part.  A contradiction since \(
      \begin{pmatrix} \av_{i}\transpose{\av}_i \\
        \Row_{i}(\transpose{\av}_i) \\ \end{pmatrix} 
      \)
      has rank $1$,
    which proves
    the first part of the lemma.
    
    There remains to prove that the map $i \mapsto j$
    is a permutation of $\{1,\dots,n\}$.
    Suppose that
	\[
	\begin{pmatrix} \transpose{\av}_i \av_{i}\\
	\Row_{i}({\av}_i)\\
	\end{pmatrix} =  \alpha_{u} \Um \begin{pmatrix} \transpose{\bv}_{u}\bv_{u} \\
	\Row_{u}({\bv}_{u}) \\
  \end{pmatrix} \Vm =  \alpha_{v} \Um \begin{pmatrix} \transpose{\bv}_{v}\bv_{v} \\
	\Row_{v}({\bv}_{v}) \\
	\end{pmatrix} \Vm
	\]
	for two distinct non-zero $\alpha_u$ and $\alpha_v$. Since $\Um$
    and $\Vm$ are non-singular, the previous equality
    implies that:
	\[
      \alpha_u
	\begin{pmatrix}  \transpose{\bv}_{u} \bv_{u} \\
	\Row_{u}({\bv}_{u}) \\
  \end{pmatrix} = \alpha_v
  \begin{pmatrix} \transpose{\bv}_{v} \bv_{v} \\
	\Row_{v}({\bv}_{v}) \end{pmatrix}.
	\]
    This contradicts the definition of $\Row_{u}({\bv}_{u})$ and
    $\Row_{v}({\bv}_{v})$
\end{proof}

Let us now prove that $(\Ccv,\Dcv)$ (codes of generator matrices $\Am$ and $\Bm$) is a  positive instance of the monomial equivalence problem if and only if $(\Ccm,\Dcm)$ is a positive instance of
\MCEP{}.

\subsection{If $(\Ccv, \Dcv)$
	is a positive instance of the monomial equivalence problem then $(\Ccm, \Dcm)$ is a positive
	instance of \MCEP{}.}

It exists
$\Sm\in\GLk$, $\Pm \in \GL{n}{\Fq}$ a permutation matrix and
$\Dm \in \GL{n}{\Fq}$ a diagonal matrix such that:
\[
\Am = \Sm\Bm\Dm\Pm.
\]
Denote by $\alpha_1, \dots, \alpha_n \in \Fq^\times$ the diagonal entries
of $\Dm$.
Then, for any $1 \leq i \leq n$, it exists $1 \leq j \leq n$
(image of $i$ by the permutation given by $\Pm$) such that,
\[
  \av_i^\top = \alpha_{j}\Sm\bv_j^\top \quad \mbox{and} \quad
  \transpose{\av}_{i} \av_i
  =\alpha_{j}^{2}\Sm\transpose{\bv}_{j}\bv_{j}\transpose{\Sm}
\]
This gives:
\[
\Ccm = \Um\Dcm\Vm, 
\]
where
\[
\Um \eqdef \begin{pmatrix}
\Sm & \mathbf{0} \\ 
\mathbf{0} & \Pm
\end{pmatrix} \in \sqMspace{k+n}{\Fq} \quad \mbox{and} \quad \Vm
\eqdef \transpose{\Sm} \in \sqMspace{k}{\Fq},
\]
which are nonsingular matrices. Therefore, $(\Ccm,\Dcm)$ is a
 positive instance of \MCEP{}.

\subsection{If $(\Ccv, \Dcv)$
	is a positive instance of \MCEP{} then $(\Ccm, \Dcm)$ is a positive
	instance of the monomial equivalence problem.}

By definition, there exist two non-singular matrices $\Um$ and $\Vm$
such that:
\[
\Ccm = \Um\Dcm\Vm.
\]
By Lemma \ref{lemma:int}, for any $1 \leq i \leq n$, it exists
$1 \leq \sigma (i) \leq n$ and $\alpha_{\sigma (i)}\in\Fq^{\times}$ such that:
\begin{align}\label{eq:linearForm}
\begin{pmatrix} \transpose{\av}_i\av_{i} \\
\Row_{i}({\av}_i)\\
\end{pmatrix} &=  \alpha_{\sigma (i)} \Um \begin{pmatrix}\transpose{\bv}_{\sigma (i)} \bv_{\sigma (i)} \\
  \Row_{\sigma (i)}({\bv}_{\sigma (i)}) \\
\end{pmatrix} \Vm \nonumber \\
&= \alpha_{\sigma (i)}\Um\begin{pmatrix}
\transpose{\bv}_{\sigma (i)} \bv_{\sigma (i)} \Vm \\
\Row_{\sigma (i)}({\bv}_{\sigma (i)}\Vm)
\end{pmatrix}.
\end{align}
The above matrices are of rank $1$ and their row spaces are generated
by $\av_{i}$ and $\bv_{i}$. Therefore as they are equal, $\av_{i}$ and
$\bv_{\sigma(i)}\Vm$ are collinear. This yields the monomial
equivalence of $\Ccv$ and $\Dcv$.

 \section*{Conclusion}

This work has presented the equivalence problem of matrix codes
endowed with the rank metric as well as some of its natural
variants. Our contribution is threefold. First, as summarized in
Figure \ref{fig:reduction2}, we have shown that the code equivalence
problem for matrix codes (\MCEP{}) is harder than the monomial
equivalence problem. This last problem has been pursued for
many years and works on it tend to show its hardness. Thus, the
equivalence problem for matrix codes is not likely to be easy and
could for instance be considered for cryptographic applications.
  \begin{figure}[h]
	\begin{center}
		\begin{tikzpicture}
			\node at (0,0.25) {\scalebox{0.9}{\textsf{Permutation}}};
			\node at (0,-0.25) {\scalebox{0.9}{\textsf{Equivalence}}};
			\draw (0,0) circle (1);
			\node at (0,-3) {\scalebox{0.9}{\textsf{Graph}}};
			\node at (0,-3.5) {\scalebox{0.9}{\textsf{Isomorphism}}};
			\draw (0,-3.25) circle (1);
			\draw[->,thick,>=latex] (0,-2.25) -- (0,-1);
			\draw (-4,-1.625) circle (1);
			\node at (-4,-1.225) {\scalebox{0.88}{\textsf{Permutation}}};
			\node at (-4,-1.625) {\scalebox{0.88}{\textsf{Equivalence}}};
			\node at (-4,-2) {\scalebox{0.88}{\textsf{with zero}}};
			\node at (-4,-2.4) {\scalebox{0.88}{\textsf{Hull}}};
			\draw[->,thick,>=latex] (-4,-0.575)  to[bend left]  (-0.757,0.707);
			\draw[->,thick,>=latex] (-4,-2.675)  to[bend right]  (-0.757,-3.927);
			\node at (4,0.25) {\scalebox{0.9}{\textsf{Monomial}}};
			\node at (4,-0.25) {\scalebox{0.9}{\textsf{Equivalence}}};
			\draw (4,0) circle (1);
			\draw[->,thick,>=latex] (0.757,0.707) to[bend left] (3.25,0.707)  ;
			\draw[<-,thick,>=latex] (0.757,-0.707) to[bend right] (3.25,-0.707);
			\node at (2,-1.5) {If $q = n^{O(1)}$};
			
			\node at (4,-3) {\scalebox{0.9}{\textsf{Matrix Code}}};
			\node at (4,-3.5) {\scalebox{0.9}{\textsf{Equivalence}}};
			\draw[dashed] (4,-3.25) circle (1);
			\draw[<-,dashed,thick,>=latex] (4,-2.25) -- (4,-1);
			\node at (5.5,-1.725) {{\em This work}};
			
\end{tikzpicture}
		\caption{Reductions around the code equivalence problem with our contribution where notation
			``$\mathsf A \longrightarrow \mathsf B$'' means that ``Problem
			$\mathsf A$ reduces to Problem $\mathsf B$ in polynomial
			time''. }
\label{fig:reduction2}
	\end{center}
\end{figure}
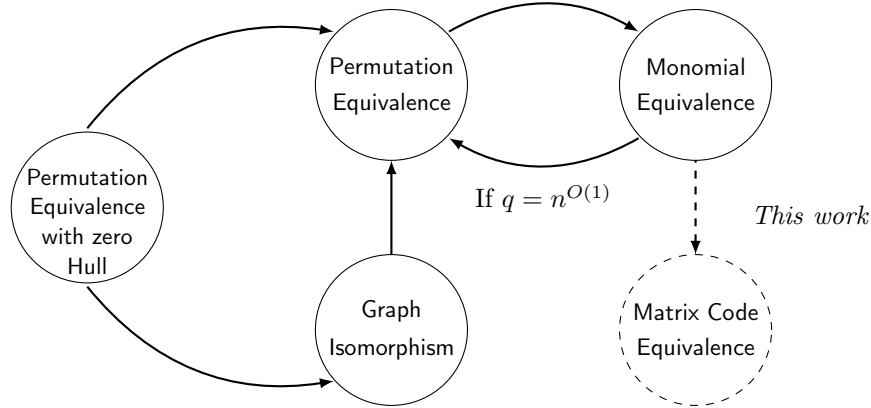
On the other hand, if one restricts instances in the equivalence
problem to $\Fqm$-linear codes (particular matrix codes that are
appealing for cryptographic applications) or only considers the
right-equivalence then it leads to an easy problem {\em in the worst case}. Using algorithms
from the mathematical field of algebras we have shown that these
problems fall down in $\mathcal{P}$ or $\mathcal{ZPP}$ depending of
the ground field size $q$.  To our opinion, this work may have
extensions that are of interest.  \newline

{\noindent \bf Classification of matrix codes.}  Our algorithm to
decide if two matrix codes are right equivalent takes advantage of the
fact that right-equivalence only shifts (multiplicatively) right
stabilizer algebras of codes by a non-singular matrix. In this way,
stabilizer algebras share, roughly speaking, the same algebraic
structure and as we have shown, by decomposing them we can decide if
two codes are right equivalent. Interestingly, our algorithm
decomposes right stabilizer algebras as ``minimal atoms''. Then it
punctures each code according to these atoms and tries to put them in
correspondence.

This may remind techniques that are used in \textit{Support Splitting}
approaches \cite{L82,S00}. It consists for two codes endowed with the
Hamming metric to puncture them and to look if they both satisfy the
same ``invariants''. In \cite{L82} is proposed to look at the weight
distribution of the punctured codes while in \cite{S00} it is this
distribution but for hulls. In this light, decomposition of right
stabilizer algebras can be seen as a good invariant for matrix
codes. By {\em good} we mean that it is an efficient tool to
solve equivalence problems. On the other hand, such invariant
remains weakly discriminant since for arbitrary codes, the
stabilizer algebra will be trivial.

This decomposition of stabilizer algebras could for instance be used
to classify matrix codes. Furthermore, stabilizer algebras may help us
to identify some properties that could be useful for decoding
algorithms. This would be particularly interesting as we know very few
families of matrix codes that we can decode efficiently. Except simple
codes \cite{SKK10} (that have a trivial structure), all matrix codes
with an efficient decoding algorithm are $\Fqm$-linear
\cite{G85,GMRZ13} which exactly corresponds to the case where codes
have a rich left stabilizer algebraic structure (that we use to decide
the equivalence of these codes).

Finally, it should be emphasized that in the present article, we made
the choice to work only over finite fields, which permits a simpler
treatment since any central simple algebra is isomorphic to a matrix
algebra in this setting (trivial Brauer group). However, some
literature exists on codes over infinite fields
\cite{augot2014generalization,
  augot2013rank,augot2018generalized,augot2020} and the similar
question of solving code equivalence over arbitrary fields makes
sense and is left as an open question.
Such a treatment will probably represent a slightly harder task while
involving non trivial central simple algebras.

\medskip

{\noindent \bf Our Reduction. } We can deduce from our reduction that
any algorithm solving a ``particular'' instance of the matrix code
equivalence problem (codes are generated by matrices of rank one)
provides an algorithm solving the monomial equivalence
problem. Therefore, this reduction gives a new manner to solve the
monomial equivalence problem. For instance we could model the matrix
equivalence problem into a system of polynomial equations that we
would solve with Gr\"obner bases techniques. This approach, thanks to
our reduction, would give a new modelling of the monomial equivalence
problem. It may be interesting to study the complexity of this
approach as Gr\"obner and to compare with the results of \cite{S17}
where permutation equivalence of codes is treating by solving a system
of polynomial equations.

\bibliographystyle{alpha}
\newcommand{\etalchar}[1]{$^{#1}$}

\end{document}